\documentclass[journal,10pt]{IEEEtran}
\makeatletter

\let\proof\@undefined
\let\endproof\@undefined
\makeatother

\usepackage{amsmath} 
\usepackage{amssymb} 
\usepackage{amsthm}
\usepackage{amsfonts}
\usepackage{mathtools}
\usepackage{mathrsfs}
\newcommand{\colvec}[2][.85]{%
  \scalebox{#1}{%
    \renewcommand{\arraystretch}{1}%
    $\begin{bmatrix}#2\end{bmatrix}$%
  }
}

\usepackage{multirow}
\usepackage{blkarray, bigstrut}

\newtheorem{prop}{Proposition}
\newtheorem{cor}{Corollary}
\newtheorem{defin}{Definition}

\newtheorem{thm}{Theorem}

\theoremstyle{remark}

\usepackage{xcolor}
\usepackage{color}
\usepackage{subfigure}
\usepackage{graphicx}
\usepackage{caption}
\usepackage{cite}
\usepackage{hyperref} 
\hypersetup{
    colorlinks,
    citecolor=black,
    filecolor=black,
    linkcolor=black,
    urlcolor=black,
    pdfauthor={},
    pdfsubject={},
    pdftitle={}
}
\newcommand{\mailtodomain}[1]{\href{mailto:#1@domain.com}{\nolinkurl{#1}}}

\newcommand\new[1]{{\color{black}{#1}}}

\usepackage[normalem]{ulem}

\usepackage[color=yellow, textsize=small, textwidth=\marginparwidth, draft]{todonotes}   % notes showed
\title{A Unified Dissertation on Bearing Rigidity Theory}

\author{Giulia Michieletto\orcidA{}, 
		Angelo Cenedese\orcidB{}
		and Daniel Zelazo\orcidC{}  
\thanks{%Support
This work was partly supported by MIUR (Italian Ministry for Education) under the initiative "Departments of Excellence" (Law 232/2016) and by University of Padova under the Visiting Scientist 2019 program and the TSTARK DEI-SEED 2020 project.
G.~Michieletto is with the Department of Engineering and Management, A.~Cenedese is with the Department of Information Engineering, both at the University of Padova, Padova, Italy ({\tt \{angelo.cenedese, giulia.michieletto\}@unipd.it});
			D.~Zelazo is with the Faculty of Aerospace Engineering, Technion-Israel Institute of Technology, Haifa, 32000, Israel ({\tt dzelazo@technion.ac.il}).
		}% <-this % stops a space
	}
	
	\definecolor{lime}{HTML}{A6CE39}
	\newcommand{\orcidicon}{%
		\begin{tikzpicture}
		\draw[lime, fill=lime] (0,0) 
		circle [radius=0.16] 
		node[white] {{\fontfamily{qag}\selectfont \tiny ID}};
		\draw[white, fill=white] (-0.0625,0.095) 
		circle [radius=0.007];
		\end{tikzpicture}
		\hspace{-2mm}}
	
	\foreach \x in {A, ..., Z}{%
		\expandafter\xdef\csname orcid\x\endcsname{\noexpand\href{https://orcid.org/\csname orcidauthor\x\endcsname}{\noexpand\orcidicon}}
	}
\newcommand{\identityMatrix}[1]{\mathbf{I}_{#1}}

\newcommand{\vectOnes}[1]{\mathbf{1}_{#1}}

\newcommand{\transposed}{\top}
\newcommand{\diag}[1]{\mathrm{diag}( #1) }

\newcommand{\skewMatrix}[1]{\left[ #1 \right]_{\times}}

\newcommand{\Immagine}[1]{\mathrm{Im} \left(#1\right)}

\newcommand{\rank}[1]{\mathrm{rk}\left(#1\right)}

\newcommand{\Ker}[1]{\mathrm{ker} \left(#1\right)}

\newcommand{\nullity}[1]{\mathrm{null}\left(#1\right)}

\newcommand{\Span}[1]{\mathrm{span}\left\lbrace#1\right\rbrace}

\newcommand{\Dim}[1]{\mathrm{dim}\left(#1\right)}

\newcommand{\worldframe}{ \mathscr{F}_W}

\newcommand{\graphG}{\mathcal{G}}

\newcommand{\graphDef}{\mathcal{G}=\left(\mathcal{V}, \mathcal{E} \right)}

\newcommand{\graphK}{\mathcal{K}}

\newcommand{\incidence}{\mathbf{E}}

\newcommand{\incidenceOut}{\mathbf{E}_
{o}}

\newcommand{\incidenceExp}{\bar{\incidence}}

\newcommand{\incidenceOutExp}{\bar{\mathbf{E}}_\text{o}}

\newcommand{\framework}[1]{\left(\mathcal{G}, #1 \right)}

\newcommand{\config}{\chi}

\newcommand{\setEquivalence}[1]{\mathcal{Q}\left(#1\right)}

\newcommand{\setCongruence}[1]{\mathcal{C}\left(#1\right)}

\newcommand{\setNeigh}[1]{\mathcal{U}\left(#1\right)}

\newcommand{\vectPos}{\mathbf{p}}

\newcommand{\vectp}[2]{\mathbf{p}_{#1#2}}

\newcommand{\vectpNorm}[2]{\bar{\mathbf{p}}_{#1#2}}

\newcommand{\bearingFunctionG}[1]{\mathbf{b}_{\mathcal{G}}\left( #1  \right)}

\newcommand{\bearingFunctionK}[1]{\mathbf{b}_{\graphK}\left( #1 \right)}

\newcommand{\bearingMatrixG}{\mathbf{B}_{\graphG}\left(\config\right)}

\newcommand{\bearingMatrixK}{\mathbf{B}_{\graphK}\left(\config\right)}

\newcommand{\infMotion}{\boldsymbol{\delta}}

\newcommand{\infMotionPos}{\boldsymbol{\delta}_{p}}

\newcommand{\infMotionAtt}{\boldsymbol{\delta}_{a}}

\newcommand{\coordRot}{\mathcal{R}_{\circlearrowleft}}

\newcommand{\trivialSpace}{\mathcal{S}_t}

\begin{document}
%%%%%%%%%%%%%%%%%%%%%%%%%%%%%%%%%%%%%%%%%%%%%%%%%%%%%%%%%%%%%%%%%%%%%%

\maketitle

%%%%%%%%%%%%%%%%%%%%%%%%%%%%%%%%%%%%%%%%%%%%%%%%%%%%%%%%%%%%%%%%%%%%%%
\begin{abstract}
This work focuses on bearing rigidity theory, namely the branch of knowledge investigating the structural properties necessary for multi-element systems to preserve the inter-unit bearings under deformations. The contributions of this work are twofold. The first one consists in the development of a general framework for the statement of the principal definitions and properties of bearing rigidity.  We show that this approach encompasses results existing in the literature, and also provides a systematic approach for studying bearing rigidity on any differential manifold in $SE(3)^n$, where $n$ is the number of agents.  The second contribution is the derivation of a general form of the rigidity matrix, a central construct in the study of rigidity theory.  We provide a necessary and sufficient condition for the infinitesimal rigidity of a bearing framework as a property of the rank of the rigidity matrix.  Finally, we present two examples of multi-agent systems not encountered in the literature and we study their rigidity properties using the developed methods. 
\end{abstract}
%%%%%%%%%%%%%%%%%%%%%%%%%%%%%%%%%%%%%%%%%%%%%%%%%%%%%%%%%%%%%%%%%%%%%%

%%%%%%%%%%%%%%%%%%%%%%%%%%%%%%%%%%%%%%%%%%%%%%%%%%%%%%%%%%%%%%%%%%%%%%
\section{Introduction}
\label{sec:intro}

According to the most general definition, rigidity theory aims at studying the \textit{stiffness} of a given system, understood as a reaction to an induced deformation. 
The origin of this branch of knowledge dates back to Euler in 1776~\cite{Euler}. In the centuries since, rigidity analysis has been extended from geometric systems to physical structures, impacting 
several and different research areas, ranging from mechanics to biology, and from robotics to chemistry~(see~\cite{thorpe1999rigidity} and the references therein). 
An outstanding result in this research thread is the work of Asimow and Roth providing the mathematical description of rigid systems of bars and joints through the notion of a \textit{framework}~\cite{asimow1978rigidity}.
This corresponds to the graph-based representation of the system (so that each vertex corresponds to a joint in the structure and each edge represents a bar connecting two elements),  jointly with a set of elements in $\mathbb{R}^d$,  $d\geq2$, describing the position of the corresponding units composing the structure~\cite{asimow1978rigidity}. 

Recently, overcoming the standard bar-and-joints frameworks,  rigidity theory has enlarged its focus toward autonomous multi-agent systems wherein the connections among the formation elements are virtual and represent sensing relations and capabilities, {and/or collective objectives}~(see, for instance,~\cite{Ahn2020, de2018formation} and the references therein). The concept of a framework has thus been redefined by considering also manifolds more complex than the ($n$-dimensional) Euclidean space. In these cases, rigidity theory turns out to be an important architectural property of many multi-agent systems where a common {global} reference frame \new{could} be unavailable but the involved devices are characterized by sensing, communication and movement capabilities. In particular, the rigidity concepts and results suitably fit applications connected to the stabilization and motion control of mobile robot formations and to sensor cooperation for localization, exploration, mapping and tracking of a target~(see, e.g.,~\cite{krick2009stabilisation,cai2015adaptive, ramazani2017rigidity,sun2018quantization, tron2015rigid, zhao2016localizability,zhao2019bearing, michieletto2019formation, Spica2016}).

\subsection{Distance vs. Bearing Rigidity} 
\label{sec:distance_vs_bearing}
{ }

Within the multi-agent systems context, rigidity properties for a given framework deal with agent interactions, either through the available sensing measurements and/or through a common objective function.  
From this perspective, the literature differentiates between \textit{distance rigidity} and \textit{bearing rigidity}. These two branches of rigidity theory in multi-agent systems arose due to the interdependence between the available sensing capabilities in a multi-agent system, and the team objective they are trying to solve.  For example, target formations specified using inter-agent distances lead to distributed controllers that require agents to gather relative position measurements~\cite{krick2009stabilisation}. Here, distance rigidity theory emerges in the convergence analysis of these systems.  On the other hand, many robotic applications employ direction or bearing-based sensors to achieve coordinated tasks.  This motivates a need for objective functions defining formations using inter-agent bearings, and consequently an extension of rigidity theory to the study of bearing-constrained frameworks~\cite{zhao2019bearing}.

The principal notions about distance rigidity are illustrated in~\cite{cai2015adaptive,eren2002framework, hendrickx2007directed,  baillieul2007combinatorial, krick2009stabilisation, zhu2009stiffness, wu2010rigidity,  williams2014evaluating, zelazo2015decentralized}. 
 These works explain how distance constraints for a framework can be summarized into a properly defined matrix whose rank determines the rigidity properties of the system analogously to the case of frameworks embedded in $\mathbb{R}^d$. In such a context,  it turns out to be useful to consider the given multi-agent system as a bar-and-joint structure where the agents are modeled as particles (joints) in $\mathbb{R}^d$, 
 and the pairs of interacting devices can be thought as being joined by bars whose lengths enforce the inter-agent distance constraints.

Bearing rigidity in $\mathbb{R}^2$ (also referred to as  \textit{parallel rigidity} in the literature~\new{\cite{eren2003sensor}}) is instead~determined by normal constraints over the directions of interacting devices, namely the edges of the graph associated to the framework, as explained in~\cite{eren2007using, bishop2011stabilization, 2012-FranchiGiordano-RigidityRd,eren2012formation}. 
These constraints entail the preservation of the angles formed between pairs of interconnected agents and the lines joining them, i.e., the inter-agent bearings.
Similar inter-agent direction constraints can be stated to access the rigidity properties of frameworks embedded in $\mathbb{R}^d$ with $d>2$, where the bearing between two agents coincides with their normalized relative direction vector~\cite{tron2015rigid,oh2014formation,  2016-ZhaoZelazo-RigidityRd, zhao2016localizability, karimian2017theory, zhao2017laman, zhao2019bearing}. In both cases, the agents are modeled as particles, 
and the necessary and sufficient condition to guarantee the rigidity properties of a given framework rests upon the rank of a matrix summarizing the involved constraints.

Dealing with a more realistic scenario, in~\cite{2014-ZelazoFranchiGiordano-RigiditySE2, 2015-ZelazoGiordanoFranchi-RigiditySE2, 2012-FranchiMasoneGrabe-RigidityR3xSO2,2016-SchianoFranchiZelazoGiordano-RigidityR3xSO2, 2016-SpicaGiordano-RigidityR3xSO2, 2016-MichielettoCenedeseFranchi-RigiditySE3} bearings are assumed to be expressed in the local frame of each agent composing the framework. This implies that each device in the group is modeled as a rigid body having a certain position and orientation with respect to a common global frame which is \new{generally} unavailable to the group. In particular, in~\cite{2014-ZelazoFranchiGiordano-RigiditySE2, 2015-ZelazoGiordanoFranchi-RigiditySE2} the attention is focused on multi-agent systems acting in the plane, in~\cite{2012-FranchiMasoneGrabe-RigidityR3xSO2, 2016-SchianoFranchiZelazoGiordano-RigidityR3xSO2, 2016-SpicaGiordano-RigidityR3xSO2} the study is extended to  the 3D space although limiting the agents attitude kinematics to rotations along only one axis, while in~\cite{2016-MichielettoCenedeseFranchi-RigiditySE3, Chen2019} fully-actuated formations are considered by assuming to deal with systems of agents having six controllable degrees of freedom (dofs). {The recent works~\cite{Schiano2018, Heintzman2018} have also established connections between rigidity and non-linear observability of dynamic systems: analogously to the former cases, the rigidity properties of the multi-agent systems are studied through the spectral analysis of a matrix that naturally appears in the observability one. \new{However, such an approach} abstracts many concepts that are taken into account in this work, as the underlying graph
of a framework and the manifold it is embedded in.}

\subsection{A  General Framework and Unified View of Bearing Rigidity}

In the past, distance rigidity has been deeply investigated from the theoretical perspective and the related multi-agent systems applications are nowadays copious, mainly focusing on formation control and localization (see, e.g.,~\cite{connelly2005generic,gortler2013affine,gortler2014characterizing, Anderson2008} for a comprehensive overview). Bearing rigidity theory, instead, has been developed only recently, gaining popularity in the last years as a multi-agent systems control strategy.

Motivated by the similarities emerging from the existing literature, the first contribution of this work proposes \textit{a unified and general framework for a bearing rigidity theory}. This allows to understand the similarities and differences of the current state-of-the-art results. To this end, {Table}~\ref{tab:2} in {Section}~\ref{sec:realization} provides a comprehensive overview of the principal features of the bearing rigidity theory for frameworks defined on different domains, while in Section \ref{sec:applications} two examples are given to demonstrate this approach for frameworks that to our knowledge do not appear in the literature.  A distinguishing feature of this contribution is the explicit consideration of frameworks over \emph{directed} graphs.  Rigidity theory for directed frameworks remains vastly unexplored, with some early results for distance constrained frameworks given in \cite{hendrickx2007directed}, and this work aims to provide a formal foundation for approaching this topic for bearing frameworks.
This unified view reveals that all notions of bearing rigidity are related through the so-called rigidity matrix.  
The second contribution of this work consists of deriving \textit{a general form for the rigidity matrix} that is predictive in the sense that its structure is completely determined by the configuration and interaction graph of the multi-agent system. We then provide a necessary and sufficient condition relating the rank of the rigidity matrix to rigidity properties of a given multi-agent system, for any agent domain.

The rest of the paper is organized as follows. {Section}~\ref{sec:PrelimNotation} summarizes basic notations and background on graph theory. 
{Section}~\ref{sec:MainDef} is devoted to the general definition and properties of  bearing rigidity theory. { The main results presented the unified framework for bearing rigidity is give in Section \ref{sec:unified_view}.  Section \ref{sec:homog_formations} discuses the results in the context of homogeneous formations, and two case studies are given in Section \ref{sec:applications}.  Finally,   {Section}~\ref{sec:colinear_cases} is devoted to a brief discussion about colinear formation cases, and concluding remarks are offered in Section \ref{sec:conclusions}. Appendix, lastly, provides the proof of an auxiliary theoretical result.}

\section{Preliminaries and Notation}
\label{sec:PrelimNotation}

A \textit{graph} is an ordered pair $\graphDef$ consisting of the vertex set $\mathcal{V} = \{ v_1 \dots v_n \}$ and the edge set $\mathcal{E}= \{ e_1 \dots e_m \} \subseteq \mathcal{V} \times \mathcal{V}$, having cardinality $|\mathcal{V}|=n$ and $|\mathcal{E}|=m$, respectively.
We distinguish between undirected, directed, and oriented graphs. An \textit{undirected graph} is a graph whereby edges have no orientation, thus $e_k = (v_i,v_j) \in \mathcal{E}$ is identical to $e_h = (v_j,v_i) \in \mathcal{E}$. Contrarily, a \textit{directed graph} is a graph whereby edges have  orientation so that the edge $e_k = (v_i,v_j) \in \mathcal{E}$ is directed from $v_i \in \mathcal{V}$ (head) to $v_j \in \mathcal{V}$ (tail). An \textit{oriented graph} is an undirected graph jointly with an orientation that is the assignment of a unique direction to each edge, hence only one directed edge ($e_k = (v_i,v_j)$ or $e_h = (v_j,v_i)$) can exist between two vertices $v_i,v_j \in \mathcal{V}$.

For any graph $\graphDef$, the corresponding \textit{complete graph} $\mathcal{K}=(\mathcal{V},\mathcal{E}_\mathcal{K})$ is the graph characterized by the same vertex set $\mathcal{V}$, while the edge set consists of all pairs of distinct vertices. Thus, for undirected graphs $|\mathcal{E}_\mathcal{K}|=n\left(n-1\right)/2$, and for directed graphs $|\mathcal{E}_\mathcal{K}|=n\left(n-1\right)$. 

For a directed/oriented graph, the \textit{incidence matrix} \mbox{$\incidence \in \mathbb{R}^{n\times m}$} is the $\{0,\pm 1\}$-matrix defined as 
\begin{equation}
\label{PrelimNotation:eq:IncidenceMatrix}
\left[\mathbf{E}  \right] _{ik}= 
\begin{cases}
- 1 & \text{if} \;\; e_k=(v_i,v_j) \in \mathcal{E} \; \text{(\textit{outgoing edge})}\\[-0.1cm]
\ \ 1 & \text{if} \;\; e_k = (v_j,v_i) \in \mathcal{E} \; \text{(\textit{ingoing edge})}\\[-0.1cm]
\ \ 0 &   \text{otherwise}, 
\end{cases}
\end{equation}
and, in a similar way, the matrix $\incidenceOut \in\mathbb{R}^{n\times m}$ is given by 
\begin{equation}
\label{PrelimNotation:eq:IncidenceMatrix_out}
\left[\incidenceOut\right]_{ik}= 
\begin{cases}
-1  & \text{if} \;\; e_k=(v_i,v_j) \in \mathcal{E} \; \text{(\textit{outgoing edge})}\\[-0.1cm]
\ \ 0 & \text{otherwise} .
\end{cases}
\end{equation}
We introduce also the matrices $\incidenceExp= \incidence \otimes \mathbf{I}_d \in \mathbb{R}^{dn\times dm}$ and $\incidenceOutExp= \incidenceOut \otimes \mathbf{I}_d \in \mathbb{R}^{dn\times dm}$, where $\otimes$ indicates the Kronecker product, $\mathbf{I}_d$ is the the $d$-dimensional identity matrix, and \mbox{$d\geq 2$} refers to the dimension of the considered space. 

The $d$-sphere, i.e., the unit sphere embedded in $\mathbb{R}^{d+1}$, is denoted as $\mathbb{S}^{d}$. 
We recall that the 1-dimensional manifold $\mathbb{S}^1$ (corresponding to the unit circle) is isomorphic to the Special Orthogonal group  $SO(2) = \{ \mathbf{R} \!\in \! \mathbb{R}^{2 \times 2 } \; \vert \; \mathbf{R}\mathbf{R}^\top \!= \! \mathbf{I}_2, \; \text{det}(\mathbf{R}) = +1\}$ which can be parametrized by a single angle $\alpha \in \left[0, 2\pi\right)$. The Special Orthogonal group $SO(3) = \{ \mathbf{R} \!\in\!  \mathbb{R}^{3 \times 3 } \; \vert \; \mathbf{R}\mathbf{R}^\top\! =\! \mathbf{I}_3,  \text{det}(\mathbf{R})=1\}$, instead,  is not isomorphic to the 2-sphere $\mathbb{S}^2$, but it holds that $\mathbb{S}^2 = {SO(3)/SO(2)}$. 
In addition, the Cartesian product $\mathbb{R}^d \times SO(d)$ corresponds to the Special Euclidean group $SE(d)$.

Considering $\mathbb{R}^d$, the vectors of its canonical basis  are indicated as 
$\mathbf{e}_i, \; i \in \{1 \ldots d\}$, and they have a one in the $i$-th entry  and zeros elsewhere. We denote with $\mathbf{0}_d \in \mathbb{R}^d$ and  $\vectOnes{d} \in \mathbb{R}^d$ 
the vectors having all the entries equal to zero and one, respectively, whereas $\mathbf{0}_{d_1 \times d_2} \in \mathbb{R}^{d_1 \times d_2}$ is the matrix having all the entries equal to zero.
Given $\mathbf{x} \in \mathbb{R}^d$, its Euclidean norm is referred as $\| \mathbf{x} \|$. We define the {\textit{orthogonal projection operator}} as 
\begin{equation}
\label{eq:orthogonal_projector}
\mathbf{P} \colon \mathbb{R}^d\rightarrow \mathbb{R}^{d\times d}, \quad\mathbf{P}\left( \mathbf{x} \right) = \mathbf{I}_{d} - \frac{\mathbf{x}}{\|\mathbf{x} \|} \frac{\mathbf{x}^{\transposed}}{\|\mathbf{x}\|},
\end{equation}
that maps any (non-zero) vector to its orthogonal complement.
Hence, $\mathbf{P}\left( \mathbf{x} \right) \mathbf{y}$ indicates the projection of $\mathbf{y}\in \mathbb{R}^d$ onto the orthogonal complement of $\mathbf{x} \in \mathbb{R}^d$. 
Given two vectors $\mathbf{x}, \mathbf{y} \in  \mathbb{R}^3$, their cross product is denoted as $\mathbf{x}\times \mathbf{y} = \skewMatrix{\mathbf{x}}\mathbf{y} = -\skewMatrix{\mathbf{y}}\mathbf{x}$, where the map $\skewMatrix{\cdot}\colon \mathbb{R}^3 \rightarrow \mathfrak{so}(3)$ associates each vector \mbox{$\mathbf{x} \in \mathbb{R}^3$} to the corresponding skew-symmetric matrix belonging to the Special Orthogonal algebra $\mathfrak{so}(3)$.
 
Given a matrix $\mathbf{A} \in \mathbb{R}^{p\times q}$, {$\mathbf{A}^\top \in \mathbb{R}^{q\times p}$ represents its \textit{transpose} and}  its \textit{null space} and \textit{image space} are denoted as $\Ker{\mathbf{A}}$ and $\Immagine{\mathbf{A}}$, respectively. The dimension of $\Immagine{\mathbf{A}}$ is indicated as $\rank{\mathbf{A}}$, whereas $\nullity{\mathbf{A}}$ stands for the nullity of the matrix, namely $\nullity{\mathbf{A}} = \text{dim}(\Ker{\mathbf{A}})$. For the well-known rank-nullity it is {$\rank{\mathbf{A}} = q -\nullity{\mathbf{A}}$. In addition, it holds that $\rank{\mathbf{A}}=\rank{\mathbf{A}^\top}$. In the rest of the paper }
we use the notation $\diag{\mathbf{A}_k} \in \mathbb{R}^{rp \times rq}$ to indicate the block diagonal matrix associated to the set $\left\lbrace \mathbf{A}_k \in \mathbb{R}^{p \times q}\right\rbrace_{k=1}^r$.  

Finally, given the function $f \colon \mathcal{X} \rightarrow \mathcal{Y}$ and the sets $\mathcal{A} \subset \mathcal{X}$, and $\mathcal{B}\subset \mathcal{Y}$, then $f\left(\mathcal{A}\right) = \left\lbrace f\left( x \right) \in \mathcal{Y}: x\in \mathcal{A} \right\rbrace$ is called the \textit{image} of $\mathcal{A}$ under $f$, and $f^{-1}\left( \mathcal{B} \right) = \left\lbrace x \in \mathcal{X}: f\left(x\right) \in \mathcal{B}\right\rbrace$ is called \textit{preimage} of $\mathcal{B}$ under $f$.

%%%%%%%%%%%%%%%%%%%%%%%%%%%%%%%%%%%%%%%%%%%%%%%%%%%%%%%%%%%%%%%%%%%%%%

%%%%%%%%%%%%%%%%%%%%%%%%%%%%%%%%%%%%%%%%%%%%%%%%%%%%%%%%%%%%%%%%%%%%%%
\section{{Bearing Rigidity: Definitions and Properties}}
\label{sec:MainDef}

In this section we introduce the main concepts related to bearing rigidity theory. 

\subsection{Framework Formation Model}
\label{sec:framework}

Consider a generic formation of $n\geq3$ agents, wherein each agent is associated to an element of the differential manifold {$\mathcal{D}_i \subseteq SE(3)$}, $i \in \{ 1 \ldots n\}$, describing its
configuration and motion constraints\footnote{{Here, we focus on real world scenarios, nonetheless the definitions and properties provided in the following are valid also for the case  $\mathcal{D}_i=\mathbb{R}^d$ with $d > 3$ discussed, e.g., in~\cite{2016-ZhaoZelazo-RigidityRd}.}}. {In detail, introducing a  global frame common (but not necessarily available) to all the agents in the group, the \textit{configuration} $\chi_i \in \mathcal{D}_i$ of the $i$-th agent coincides with its position when modeled as a {particle}, or with the pair of its position and (partial/full) attitude, i.e., with its (partial/full) pose, when the rigid body model is assumed.} 
{We now introduce the notion of a framework.}

\begin{defin}[Framework in $\bar{\mathcal{D}}$]
	\label{MainDef:def:Framework}
	A \emph{framework} embedded in the {product} differential manifold  $\bar{\mathcal{D}}= \prod_{i=1}^n \mathcal{D}_i {\subseteq SE(3)^n}$  is an ordered pair $\left(\mathcal{G}, \config \right)$ consisting of a (directed or undirected) graph $\mathcal{G}=\left(\mathcal{V}, \mathcal{E} \right)$ with $\vert \mathcal{V}\vert=n$, and a map $\chi : \mathcal{V} \to \bar{\mathcal{D}}$ such that $v_i \mapsto \chi(v_i):=\chi_i \in \mathcal{D}_i$.  We refer to $\config = \left(\chi_1 \dots \chi_n\right) \in  \bar{\mathcal{D}}$ as the \emph{formation configuration}.
\end{defin}

The framework model characterizes a formation in terms of the {agent configuration}, where{ $\chi$} can be thought of as an embedding of the graph into $\bar{\mathcal{D}}$, and the agents are associated with nodes in the graph.  {In the study of bearing rigidity, we are interested in the bearing vector between pairs of agents that are connected by an edge in $\mathcal{G}$.} Note that $\mathcal{G}$ can be directed or undirected. In rigidity theory, it is typically assumed that the graph is not time-varying, and we adopt this assumption here.

We introduce also the following definitions on formations.

\begin{defin}[Non-Colinear Formation]
\label{MainDef:def:colinear}
An $n$-agent formation modeled as a framework $\left(\mathcal{G}, \config \right)$ in $\bar{\mathcal{D}}$ is \emph{non-colinear} if all  agents are {in distinct positions} and do not lie along the same line in the global frame.
\end{defin}

\begin{defin}[Homogeneous Formation]
\label{MainDef:def:Homogeneous}
	An $n$-agent formation modeled as a framework $\left(\mathcal{G}, \config \right)$ in $\bar{\mathcal{D}}$ is \emph{homogeneous} if  $\mathcal{D}_i = \mathcal{D}$ $\forall i \in \{1 \ldots n\}$, hence $\bar{\mathcal{D}} = \mathcal{D}^n$. Otherwise, the formation is \emph{heterogeneous}. 
\end{defin}

{Hereafter, we focus on non-colinear formations, albeit the colinear case is discussed in Section~\ref{sec:colinear_cases}.  Note that, for a non-colinear homogeneous $n$-agent formation, the $(n \times d)$ matrix %of the coordinates
 describing the agents position in $\mathbb{R}^d$ is of rank greater than 1.
 
 Although the stated assumptions regard the agents state and motion constraints, for a given formation the bearing rigidity properties are related to the bearing vector between neighboring agents.} 
According to the framework model, every edge $e_k=e_{ij}=\left(v_i, v_j \right)\in \mathcal{E}$ ($\vert\mathcal{E} \vert = m$) represents a \textit{bearing measurement} $\mathbf{b}_k=\mathbf{b}_{ij}$ defined in the differential manifold {$\mathcal{M}_k\subseteq \mathbb{S}^2$} and recovered by the $i$-th agent which is able to sense the $j$-th agent, $i,j \in \{1 \ldots n\}, i \neq j$. The \textit{bearing measurements domain} can now be expressed as $\bar{\mathcal{M}}= \prod_{k=1}^m \mathcal{M}_k{\subseteq \mathbb{S}^{2m}}$. For homogeneous formations, it is $\bar{\mathcal{M}}= \mathcal{M}^m$ with $\mathcal{M}_k = \mathcal{M}, \forall k \in \{1 \ldots m\}$.
The available measurements can be expressed in the global frame or {in the local frame in-built with  each agent/node (and, thus, defined according to $\mathcal{D}_i$)}. However, in both cases, these are related to the framework configuration according to the following  
definition where an arbitrary edge labeling is introduced.

\begin{defin}[{Bearing Function}]
\label{MainDef:def:BR_function}
	Given an $n$-agent formation modeled as a framework $\left(\mathcal{G}, \config \right)$ in $\bar{\mathcal{D}}$, the {\emph{bearing function}} is the map {$\mathbf{b}_{\graphG} \colon \bar{\mathcal{D}} \rightarrow \bar{\mathcal{M}}$}
	associating the {formation configuration} $\config \in \bar{\mathcal{D}}$ to the vector $\bearingFunctionG{\config}=\colvec{ \mathbf{b}_{1}^{\transposed} \;\; \dots \;\; \mathbf{b}_{m}^{\transposed} }^\transposed \in \bar{\mathcal{M}}$, stacking all the available bearing measurements.	 
\end{defin}

Observe that the bearing function determines the \emph{shape} of the formation in terms of relative pose among all the agents.  One of the central questions in  bearing rigidity theory is if a given formation with its bearing function \emph{uniquely} defines the shape.  This will explored in the sequel.

Hereafter, the framework model is adopted to refer to an $n$-agent formation (implying $n\geq3$) and the two concepts (framework and formation) are assumed to be equivalent.

%%%%%%%%%%%%%%%%%%%
%%%%%%%%%%%%%%%%%%%
\subsection{Rigidity Properties of Static Frameworks}
\label{static_rigidity}
 
{Definition}~\ref{MainDef:def:BR_function} allows to introduce the first two notions related to bearing rigidity theory, namely the equivalence and the congruence of different frameworks.  

\begin{defin}[Bearing Equivalence]
	\label{MainDef:def:BE}
 		Two frameworks $\framework{\config}$ and $\framework{\config'}$ are \emph{bearing equivalent (BE)} if $\bearingFunctionG{\config} = \bearingFunctionG{\config'}$.
\end{defin}
\begin{defin}[Bearing Congruence]
	\label{MainDef:def:BC}
		Two frameworks $\framework{\config}$ and $\framework{\config'}$ are \emph{ bearing congruent (BC)} if $\bearingFunctionK{\config} = \bearingFunctionK{\config'}$, where $\mathcal{K}$ is the complete graph {associated to $\mathcal{G}$}.
\end{defin}

Accounting for the preimage under the bearing function, the set $\setEquivalence{\config} = \mathbf{b}_{\mathcal{G}}^{-1} \left( \mathbf{b}_{\mathcal{G}}\left(\config\right)\right) \subseteq \bar{\mathcal{D}}$ includes all the {formation configurations} $\config' \in \bar{\mathcal{D}}$ such that $\framework{\config'}$ is BE to $\framework{\config}$, while the set $\setCongruence{\config} =\mathbf{b}_{\mathcal{K}}^{-1} \left( \mathbf{b}_{\mathcal{K}}\left(\config\right)\right) \subseteq \bar{\mathcal{D}}$ contains all the {formation configurations} $\config'\in \bar{\mathcal{D}}$ such that $\framework{\config'} $ is BC to $\framework{\config}$. Trivially, it follows that $\setCongruence{\config}\subseteq \setEquivalence{\config}$.

The definition of these sets allows to introduce the (local and global) property of bearing rigidity.

\begin{defin}[Bearing Rigidity in $\bar{\mathcal{D}}$]
\label{MainDef:def:BR}	
	A framework $\framework{\config}$ is (locally) \emph{bearing rigid (BR)} in $\bar{\mathcal{D}}$ if there exists a neighborhood $\setNeigh{\config} \subseteq \bar{\mathcal{D}}$ of $\config$  such that 
		\begin{equation}
			\label{MainDef:eq:BR_condition}
			\setEquivalence{\config} \cap \setNeigh{\config} = \setCongruence{\config} \cap \setNeigh{\config}.
		\end{equation}
\end{defin}

\begin{defin}[Global Bearing Rigidity in $\bar{\mathcal{D}}$]
\label{MainDef:def:GBR}
	A framework $\left(\mathcal{G}, \config \right)$ is \emph{globally bearing rigid (GBR)} in $\bar{\mathcal{D}}$ if every framework which is BE to $\left( \mathcal{G}, \config \right)$ is also BC to $\left( \mathcal{G}, \config \right)$, \new{i.e.,} 
	%or equivalently if 
	$\setEquivalence{\config} = \setCongruence{\config}$.
\end{defin}

Figure~\ref{MainDef:fig:BR_Inclusions} provides a graphical interpretation of condition~\eqref{MainDef:eq:BR_condition} highlighting the relation between the sets $\setEquivalence{\config}, \setCongruence{\config}$ and $\setNeigh{\config}$. The {existence  of  a  neighborhood} in the configurations space is {not present} in {Definition}~\ref{MainDef:def:GBR} of global bearing rigidity. As a consequence, this property results to be stronger than the previous one as proved in the next theorem.

\begin{figure}[t!]
\centering
	\includegraphics[trim={0 2pt 0 2pt}, clip, width=0.75\columnwidth]{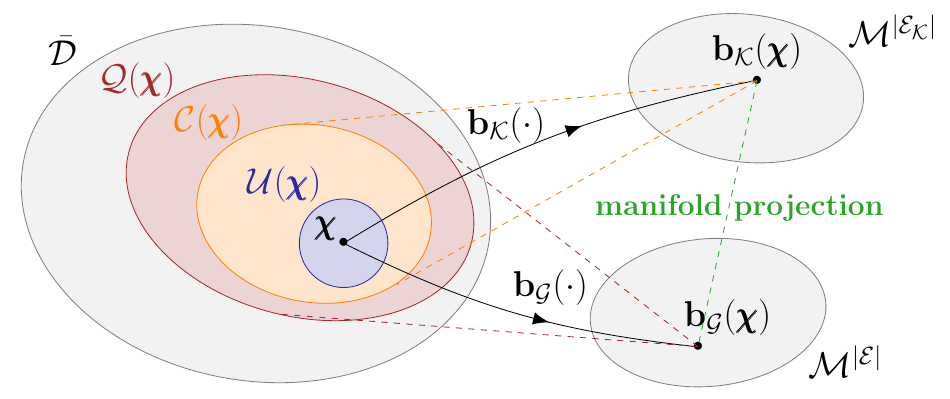}
	\caption{Graphical interpretation of condition~\eqref{MainDef:eq:BR_condition}.}
	\label{MainDef:fig:BR_Inclusions}
	\vspace{-0.2cm}
\end{figure}

{
\begin{prop}
\label{prop:GBR_BR}
	\textcolor{black}{A GBR framework $\framework{\config}$ is also BR.} 
\end{prop}}
\begin{proof}
For a GBR framework $\framework{\config}$, it holds that $\setEquivalence{\config} = \setCongruence{\config}$. Consequently, condition~\eqref{MainDef:eq:BR_condition} is valid for $\setNeigh{\config}= \bar{\mathcal{D}}$ demonstrating that the framework is BR. 
\end{proof}

\subsection{{Rigidity Properties of Dynamic Frameworks}}
\label{sec:dynamic_rigidity}

All the properties previously defined concern rigidity for \textit{static} frameworks. In real-world scenarios, however, agents are generally 
able to move with or within a formation, and, thus, to change their configuration.  Here, the motion constraints are captured by the configuration space $\mathcal{D}_i$ of each agent.  
For this reason, in this section we assume to deal with \textit{dynamic} agent formations, namely formations modeled as frameworks $\framework{\config}$ with fixed sensing graph $\mathcal{G}$ and $\config$ subject to a deformation implied by the variation of the single agent configuration. Formally, we consider the motion of an agent along a curve in its configuration space parameterized by a variable $t \in [0,1]$.  Thus, we have that $\chi_i=\chi_i(t) \in \mathcal{D}_i$ and $\config=\config(t)=\{\chi_1(t) \ldots \chi_n(t) \} \in \bar{\mathcal{D}}$. 

\new{In this context, we introduce the \textit{variation} $\infMotion_i$ defined in the $i$-th agent \textit{variations domain} $\mathcal{I}_i$, depending on the $i$-th agent {motion constraints}. In particular, we assume that  $\frac{d\chi_i(t)}{dt}= f_i(\chi_i(t),\infMotion_i)$  where 
$f_i \colon {\mathcal{D}}_i  \times {\mathcal{I}}_i \rightarrow {\mathcal{D}}_i$ is a continuous function.} Accounting for the whole formation, we can stack all $\infMotion_i$ into a vector $\infMotion \in  \bar{\mathcal{I}}$  where $\bar{\mathcal{I}}= \prod_{i=1}^n \mathcal{I}_i$  is the \textit{variations domain}. Note that, for homogeneous formations, we have $\mathcal{I}_i = \mathcal{I}$, and thus $\bar{\mathcal{I}}= \mathcal{I}^n$.
Hereafter, we interpret the variable $t$ as a \emph{time} variable and $\infMotion$ as the vector of commands acting on the formation to attain the desired evolution from an initial formation configuration $\chi(0)$ to a final one $\chi(1)$.

Given this setup, our aim is to identify the conditions
under which a given dynamic formation deforms while maintaining its rigidity features, i.e., preserving the existing bearings among the agents over time. 
The
relation between ${\infMotion}$ and  the derivative of the {bearing function}, clarified in the next definition, constitutes the starting point for the study of the formation rigidity properties.

\begin{defin}[Bearing Rigidity Matrix]
	\label{MainDef:def:BR_matrix}
	For a given (dynamic) framework $\left(\mathcal{G}, \chi(t) \right)$, the \emph{bearing rigidity matrix} is the matrix $ \mathbf{B}_{\mathcal{G}}(\chi(t)) $ that satisfies the relation
	\begin{equation}
	\label{MainDef:eq:BR_matrix}
	\dot{\mathbf{b}}_{\mathcal{G}}({\config(t)}) =
		\frac{d}{dt}\mathbf{b}_{\mathcal{G}}(\chi(t))= \mathbf{B}_{\mathcal{G}}(\chi(t)) {\infMotion}.
	\end{equation}
\end{defin}
{Indeed, $\mathbf{B}_{\mathcal{G}}(\chi(t))$ can be interpreted as a Jacobian matrix within a differential geometry perspective,} 
whose dimensions depend on the spaces $\bar{\mathcal{M}}$ and $\bar{\mathcal{I}}$.
Nevertheless, one can observe that the null space of $\mathbf{B}_{\mathcal{G}}(\chi(t)) $ always identifies all the (first-order) deformations of $\chi(t)$ that keep the bearing measurements unchanged, following from the Taylor series expansion of the {bearing function}. 
From a physical perspective, such variations of $(\mathcal{G},\chi(t))$ can be interpreted as sets of command inputs to provide to the agents to drive the formation from {an} initial configuration $\config(0)$ to a final one $\config(1)$ belonging to the set $\setEquivalence{\config(0)}$ of equivalent formation configurations.  Note that the existence of such paths further implies a smooth mapping from $\config(0)$ to $\config(1)$, and in this sense, $\infMotion$ can be interpreted as an instantaneous velocity \cite{asimow1979rigidity}. 

\begin{defin}[Infinitesimal Variation]
	\label{MainDef:def:Variation_Infinitesimal} 
	{For a given (dynamic) framework $\left(\mathcal{G},\chi(t) \right)$, a variation $\infMotion \in \bar{\mathcal{I}}$ is  \emph{infinitesimal} if and only if $\infMotion \in\Ker{\mathbf{B}_{\mathcal{G}}(\chi(t))}$.}
\end{defin}

It directly results from~\eqref{MainDef:eq:BR_matrix} that an infinitesimal variation preserves the bearing measurements among all interacting agents.
For a given $(\mathcal{G},\chi(t))$, there may be many infinitesimal variations.  However, there exist infinitesimal variations that hold for \emph{any} graph.  This follows from the next result.
\begin{thm}
	\label{MainDef:thm:Inclusion_KerB_K_KerB_G}
	Given a dynamic framework $(\mathcal{G},\chi(t))$, and denoting as $\mathcal{K}$ the complete graph {associated to~$\mathcal{G}$}, it holds that $\Ker{\mathbf{B}_{\mathcal{K}}(\chi(t))} \subseteq \Ker{\mathbf{B}_{\mathcal{G}}(\chi(t))}$.
\end{thm}
\begin{proof}
Since each edge of the graph $\mathcal{G}$ belongs to the graph $\mathcal{K}$, the equations set {$\mathbf{B}_{\mathcal{G}}(\chi(t)){\infMotion} = \mathbf{0}_\mu$} constitutes a subset of the equations set {$\mathbf{B}_{\mathcal{K}}(\chi(t)) {\infMotion} = \mathbf{0}_\mu$} {with $\mu$ depending on $\bar{\mathcal{M}
}$}. Then ${\infMotion} \in \Ker{\mathbf{B}_{\mathcal{K}}(\chi(t))}$ implies ${\infMotion} \in \Ker{\mathbf{B}_{\mathcal{G}}(\chi(t))}$.
\end{proof}

In light of {Theorem}~\ref{MainDef:thm:Inclusion_KerB_K_KerB_G}, we introduce the notion of \emph{trivial variations} by considering the infinitesimal variations related to the complete graph $\mathcal{K}$. These ensure the measurements preservation for each pair of nodes in the formation, i.e., the formation shape preservation in terms of relative poses among the agents. 

\begin{defin}[Trivial Variation]
	\label{MainDef:def:Variation_Trivial}
	For a given (dynamic) framework $\left(\mathcal{G}, \chi(t) \right)$, a  variation $\infMotion \in \bar{\mathcal{I}}$ is \emph{trivial} if and only if $\infMotion \in \Ker{\mathbf{B}_{\mathcal{K}}(\chi(t))}$, where $\mathbf{B}_{\mathcal{K}}(\chi(t) )$ is the bearing rigidity matrix computed for the complete graph $\mathcal{K}$ associated to $\mathcal{G}$.
\end{defin}

Theorem~\ref{MainDef:thm:Inclusion_KerB_K_KerB_G} 
is fundamental for the next definition that constitutes a key concept in rigidity theory.

\begin{defin}[Infinitesimal Bearing Rigidity in $\bar{\mathcal{D}}$]
	\label{MainDef:def:IBR}
	A (dynamic) framework $\left(\mathcal{G}, \chi(t) \right)$ is  \emph{infinitesimally bearing rigid (IBR)} in $\bar{\mathcal{D}}$ if 
$	\Ker{\mathbf{B}_{\mathcal{G}}(\chi(t))} = \Ker{\mathbf{B}_{\mathcal{K}}(\chi(t))}$.
	Otherwise, it is \emph{infinitesimally bearing flexible (IBF)}. 
\end{defin}

\new{From a physical point of view, }a framework $\left(\mathcal{G},\chi(t) \right)$ is IBR if its \textit{trivial variations set} $\mathcal{S}_t :=  \Ker{\mathbf{B}_{\mathcal{K}}(\chi(t))} $ coincides with its \textit{infinitesimal variations set} $\mathcal{S}_i := \Ker{\mathbf{B}_{\mathcal{G}}(\chi(t))}$.
\new{A variation in the set $\mathcal{S}_i$ ($\mathcal{S}_t$)
induces a deformation of the configuration $\chi(0)$ into $\chi(1)$ that is bearing equivalent (congruent) to the initial one.} %, i.e., $\chi(1) \in \mathcal{C}(\chi(0))$.}
\new{Thus, for an IBF framework, there exists at least a variation that {deforms} the configuration $\config(0)$ to $\chi(1)\in \setEquivalence{\config(0)}\setminus\setCongruence{\config(0)}$.} 

In the rest of the paper, we  limit our analysis to the dynamic framework case, and whenever it is possible, the time dependency is dropped out to simplify the notation.

%%%%%%%%%%%%%%%%%%%%%%%%%%%%%%%%%%%%%%%%%%%%%%%%%%%%%%%%%%%%%%%%%%%%%%
\section{Unified Rigidity Theory}
\label{sec:unified_view}

%%%%%%

Many of the existing works on the bearing rigidity begin their analysis with a statement of the agent configuration space.  The rigidity results are then derived, not in generic terms, but explicitly as a function of the chosen configuration space. A consequence of this approach is the need to rederive and redefine rigidity concepts. 
In this section, we show that stemming from the general setup given in Section \ref{sec:MainDef}, we are able to unify the study of bearing rigidity that holds for any $\bar{\mathcal{D}}$ inside $SE(3)^n$. 
The main realization is that any agent can be interpreted as a rigid body acting in 3D space with constraints on its motions (i.e., constrained to move on a differential manifold inside $SE(3)^n$). 
This approach leads to a general and constructive form for the rigidity matrix and a necessary and sufficient condition relating infinitesimal bearing rigidity to the rank of this matrix.

% %%%%%%%%%%%%%%%%%%%%%%%%%%%%%%%%%%%%%%%%%%%%%%%%%%%%%%%%%%%%%%%%%%%%%%
% \subsection{On the Structure of the Rigidity Matrix}
% \label{sec:structure_rigidity}

% %%%%%%%%%%%%%%%%%%%%%%%%%%%%%%%%%%%%%%%%%%%%%%%%%%%%%%%%%%%%%%%%%%%%%%

We consider a formation where each $i$-th agent, $i \in\{ 1 \ldots n\}$,  is associated to a set of bearing vectors related to its neighbors defined by the graph $\mathcal{G}$. An agent can vary its configuration  within $\mathcal{D}_i$ consisting of  $c_i^t \in \mathbb{N}$ translational dofs (tdofs) and $c_i^r \in \mathbb{N}$ rotational dofs (rdofs), where $c_i = c_i^t+c_i^r$ is the dimension of the differential manifold $\mathcal{D}_i$.
In this work we consider $\mathcal{D}_i$ in $SE(3)$, so $c_i^t$ and $c_i^r$ are limited in $[0, 3]$.

%We observe that, 
Independently of $c_i$, each agent in the group can be modeled as a rigid body and associated to a \textit{local reference frame} $\mathscr{F}_i$ whose origin $O_i$ coincides with the agent center of mass. At each time instant $t \geq 0$, it is thus possible to describe the {pose} of the agent in the 3D space through the pair $(\mathbf{p}_i(t),\mathbf{R}_i(t))\in SE(3)$, where~the~vector $\mathbf{p}_i(t)= \colvec{p_i^x(t) &p_i^y(t) & p_i^z(t)}^\top \in \mathbb{R}^3$ identifies the position of $O_i$ in the {global} frame $\worldframe$ and the matrix $\mathbf{R}_i(t)\in SO(3)$ defines the orientation of $\mathscr{F}_i$ {with respect to} $\worldframe$. 
In particular, supposing that the (unit) vectors $\mathbf{e}_h \in \mathbb{S}^2$, $h \in \{1,2,3\}$ identify the axes of the global frame, we assume that {$\mathbf{R}_i(t)= \mathbf{R}\left(\{\theta_{i,h}(t),  \mathbf{e}_h\}_{h=1}^{3}\right)$ meaning that $\mathbf{R}_i(t)$ results from the composition of three consecutive rotations, each of them performed around $\mathbf{e}_h$ of an angle $\theta_{i,h}(t)$, according to a suitable sequence.}\footnote{{This reasoning remains valid for any representation of 3D rotations.}}
The parameter $c_i$, on the other hand, is intrinsically related to the dimension of  $ \mathcal{D}_i$, and $\chi_i$  may not necessarily coincide with the whole pair $(\mathbf{p}_i(t),\mathbf{R}_i(t))$. Specifically, when $c_i<6$, the agent can vary its pose in 3D space only \textit{partially}.

\begin{table}[t!]
% \begin{center}
\centering
\resizebox{0.49\textwidth}{!}{
\renewcommand{\arraystretch}{1.1}
\begin{tabular}{|c|c|c|c|c|}
\hline
{$\mathcal{D}$} & $\mathbf{p}_i$ & $\mathbf{R}_i$ & {$\mathbf{U}_{ij}$} & {$\mathbf{V}_{ij}$} \\ \hline\hline
$SE(3)$ & $\colvec{p_i^x &p_i^y & p_i^z}^\top$ & {$ \mathbf{R}\left(\{\theta_{i,h}(t),  \mathbf{e}_h\}_{h=1}^{3}\right)$} & $\mathbf{I}_3$ & $\mathbf{I}_3$ \\ 
 \hline \hline 
 $\mathbb{R}^3\times \mathbb{S}^1$ & $\colvec{p_i^x &p_i^y & p_i^z}^\top$ &  $\mathbf{R}\left({\theta_i}(t),\mathbf{v}\right), \mathbf{v}={\footnotesize \displaystyle {\sum_{h=1}^3}} v_h \mathbf{e}_h$ & $\mathbf{I}_3$ & $[\mathbf{0}_{3\times 2} \; \mathbf{v}]$ \\ 
$\mathbb{R}^2\times \mathbb{S}^1$ & $\colvec{p_i^x &p_i^y & 0}^\top$ &  $\mathbf{R}\left({\theta_i}(t),\mathbf{e}_3\right)$ & $\colvec{\mathbf{e}_1 \; \mathbf{e}_2 \; \mathbf{0}_3}$ & $[\mathbf{0}_{3\times 2} \; \mathbf{e}_3]$ \\  
 \hline\hline 
$\mathbb{R}^3$ & $\colvec{p_i^x &p_i^y & p_i^z}^\top$ &  $\mathbf{I}_3$ & $\mathbf{I}_3$ & $\mathbf{0}_{3 \times 3}$ \\ 
$\mathbb{R}^2$ & $\colvec{p_i^x &p_i^y & 0}^\top$ & $\mathbf{I}_3$ & $\colvec{\mathbf{e}_1 \; \mathbf{e}_2 \; \mathbf{0}_3}$ & $\mathbf{0}_{3 \times 3}$ \\ 
\hline
\end{tabular}
}
%\end{center}
\caption{Particularization of the structure of the {extended bearing matrix} in~\eqref{eq:General_Bearing_Matrix} for the {differential manifolds} considered in {{Section}~\ref{sec:realization}}.}
\label{tab:generic_rigidity_matrix}
\vspace{-0.2cm}
\end{table}

In light of Section~\ref{sec:MainDef}, the described formation can be modeled as a framework in  $\bar{\mathcal{D}}\subseteq SE(3)^n$. 
Under the %(non-restrictive) 
assumption that agents do not have access to a global frame,  $\graphG$  is a \textit{directed} graph encoding that bearings are inherently expressed in the local frames and are not necessarily reciprocal between pair of agents. %
Hence, the directed edge $e_k = \left(v_i, v_j \right) \in \mathcal{E}$ refers to the bearing of the $j$-th agent obtained by the $i$-th agent. This can be expressed in terms of the relative position and orientation of the agents in $\mathscr{F}_W$, namely
\begin{align}
\label{eq:Bearing_Measure}
\mathbf{b}_k(t) = \mathbf{b}_{ij}(t) &=  \mathbf{R}_{i}^{\transposed}(t) s_{ij}(t) \vectp{i}{j}(t) =   \mathbf{R}_{i}^{\transposed}(t) \vectpNorm{i}{j}(t),
\end{align}
where $\vectp{i}{j}(t)=\mathbf{p}_j(t)- \mathbf{p}_i(t)\in \mathbb{R}^{d}$ is the relative position vector, and $s_{ij}(t)=\| \vectp{i}{j}(t)\|^{-1}\in\mathbb{R}$ is the inverse of the relative distance between the $i$-th and $j$-th agent.

To treat in a unified way multiple domains, we can consider the embedding\footnote{{Hereafter, a superscript $+$ is used to highlight the vectors defined in the lifted spaces.}} of each ${\cal D}_i$ into the $SE(3)$ manifold, thus considering the given formation as a framework $(\mathcal{G},\chi(t))$ in $SE(3)^n$.
Observing~\eqref{eq:Bearing_Measure}, we embed $ \bar{\mathcal{M}}$ into $\mathbb{S}^{2m}$ and, according to Definition~\ref{MainDef:def:BR_function}, the {bearing function} can be expressed as
\begin{equation}
\label{eq:Bearing_Function}
    {\mathbf{b}_\mathcal{G}^+(\chi(t))} = \diag{s_{ij}(t)\mathbf{R}_{i}^{\transposed}(t)}\incidenceExp^{\transposed} \vectPos(t) \in \mathbb{S}^{2m},
\end{equation}
where $\mathbf{p}(t) = \colvec{ \mathbf{p}_1^\top(t)  \dots  \mathbf{p}_n^\top(t) }^\top \in \mathbb{R}^{3n}$ stacks all the agent position vectors. 
Consistently, the variations domain $\bar{\cal I}$ can be embedded in $\mathbb{R}^{6n}$ and \new{thus the vector  $\infMotion$ can be substituted by}
\begin{align}
\label{eq:General_Variation}
{\infMotion^+} = \colvec{\infMotionPos^{\transposed}  & \infMotionAtt^{\transposed}
}^{\transposed} \in \mathbb{R}^{6n},
\end{align}
\new{where $\infMotionPos \in \mathbb{R}^{3n}$ and $\infMotionAtt \in \mathbb{R}^{3n}$ are defined by padding with zeros the corresponding components of $\infMotion$ related to the possible variations of the agents position and attitude, respectively.}

{
This allows to introduce the following definition.
\begin{defin}[Extended Bearing Rigidity Matrix]
	\label{def:ex_BR_matrix}
	For a given framework $\left(\mathcal{G}, \chi(t) \right)$ embedded in $SE(3)^n$, the \emph{extended bearing rigidity matrix} is the matrix $\mathbf{B}^+_\mathcal{G}(\chi(t))\in \mathbb{R}^{3m \times 6n}$ that satisfies the relation
	\begin{equation}
	\label{MainDef:eq:ex_BR_matrix}
	\dot{\mathbf{b}}_{\mathcal{G}}^+({\config(t)}) =
		\frac{d}{dt}\mathbf{b}_{\mathcal{G}}^+(\chi(t))= \mathbf{B}_{\mathcal{G}}^+(\chi(t)) {\infMotion}^+.
	\end{equation}
\end{defin}
}

\new{Note that since $\mathbf{b}_\mathcal{G}(\chi(t))^+$ can be interpreted as the zero-padded version of the vector $\mathbf{b}_\mathcal{G}(\chi(t))$ in~\eqref{eq:Bearing_Measure}, the relation~\eqref{MainDef:eq:ex_BR_matrix} 
%provided in Definition~\ref{def:ex_BR_matrix} 
corresponds to~\eqref{MainDef:eq:BR_matrix} when accounting for the embedding of $\mathcal{D}_i$ in $SE(3)$. The consistency between $\mathbf{B}_{\mathcal{G}}(\chi(t))$ and $\mathbf{B}_{\mathcal{G}}(\chi(t))^+$ is guaranteed by the emergence of zero rows in the latter. Along the same line,  we also observe that~\eqref{eq:General_Variation} induces a partition of $\mathbf{B}^+_\mathcal{G}(\chi(t))$  into two blocks, distinguishing between the bearing variations due to the variations of the agents position and orientation, as formalized in the next proposition.}

\begin{prop}
\label{prop:rigidity_matrix} 
The extended bearing rigidity matrix in Definition~\ref{def:ex_BR_matrix} can be expressed as
\begin{align}
\label{eq:General_Bearing_Matrix}
\mathbf{B}^+_{\mathcal{G}}(\chi(t)) &= \colvec{ \mathbf{D}_p(t)\mathbf{U} \bar{\mathbf{E}}^\top  & \mathbf{D}_{{a}}(t) \mathbf{V}\bar{\mathbf{E}}_{o}^\top} \in \mathbb{R}^{3m \times 6n},
\end{align}
where
\begin{itemize}
    \item  $\mathbf{D}_p(t), \mathbf{D}_{{a}}(t) \in \mathbb{R}^{3m \times 3m}$ are derived from the orthogonal projections of relative position and attitude 
\begin{align}
 \mathbf{D}_p(t)&= \diag{s_{ij}(t) \mathbf{R}_i^\top(t)\mathbf{P}\left(\vectpNorm{i}{j}(t) \right) },
  \label{eq:General_Bearing_Matrix_middle} \\ 
 \mathbf{D}_{{a}}(t)&=  -\diag{ \mathbf{R}_i^{\transposed}(t)\skewMatrix{\vectpNorm{i}{j}(t)}},
 \label{eq:General_Bearing_Matrix_end}
\end{align}
\item 
$\mathbf{U}=\diag{\mathbf{U}_{ij}}\in \mathbb{R}^{3m \times 3m}$ and $\mathbf{V}=\diag{\mathbf{V}_{ij}}\in \mathbb{R}^{3m \times 3m}$ take into account the (time-invariant) matrices  $\mathbf{U}_{ij}, \mathbf{V}_{ij} \in \mathbb{R}^{3 \times 3}$  defining, respectively, the translational directions of the bearing measurement $\mathbf{b}_{ij}$ and $i$-th and $j$-th agents rotation directions in the 3D space with respect to the $i$-th agent frame;
\item 
$\incidenceExp, \incidenceOutExp\in \mathbb{R}^{3n\times 3m}$ are derived from the (time-invariant) incidence matrix of the graph $\graphG$.
\end{itemize}
\end{prop}

\begin{proof}
According to Definition~\ref{def:ex_BR_matrix}, 
the extended bearing rigidity matrix  $\mathbf{B}^+_\mathcal{G}(\chi(t))$ has to map 
the vector $\infMotion^+\in\mathbb{R}^{6n}$ to the derivative of the bearing function. %~\eqref{eq:Bearing_Function}. 
\new{By applying the product rule to~\eqref{eq:Bearing_Function}, 
%and accounting for the rotation kinematics $\dot{\mathbf{R}}_i= \mathbf{R}_i^\top  \skewMatrix{\boldsymbol{\omega}_i}$ where $\boldsymbol{\omega}_i \in \mathbb{R}^3$ identifies the angular velocity of the $i$-th agent express in $\mathscr{F}_i$, 
it follows that
\begin{align}
\label{eq:b_plus_derivative}
 \dot{\mathbf{b}}_{\mathcal{G}}^+({\config(t)}) &=  \left(\frac{d}{dt}\diag{s_{ij}(t)}\right) \diag{\mathbf{R}_{i}^{\transposed}(t)}\incidenceExp^{\transposed} \vectPos(t) \\
 &+ \diag{s_{ij}(t)} \left(\frac{d}{dt}\diag{\mathbf{R}_{i}^{\transposed}(t)}\right) \incidenceExp^{\transposed} \vectPos(t)  \nonumber\\
 &+ \diag{s_{ij}(t)} \diag{\mathbf{R}_{i}^{\transposed}(t)}  \incidenceExp^{\transposed} \frac{d}{dt}\vectPos(t). \nonumber
\end{align}
In light of~\eqref{eq:General_Variation}, it is possible to distinguish the measurements variations induced by the variations of the agents position $\infMotionPos$ or attitude $\infMotionAtt$:  the former contribution refers to the first and third addendum in~\eqref{eq:b_plus_derivative}, while the latter contribution is related to the second one.
Accounting for the agents dynamics when embedded in $SE(3)$,~\eqref{eq:b_plus_derivative} can then be rearranged as
\begin{align}
\label{eq:b_plus_derivative_2}
 \dot{\mathbf{b}}_{\mathcal{G}}^+({\config(t)}) &=  \mathbf{D}_p(t)\mathbf{U} \bar{\mathbf{E}}^\top \infMotionPos  + \mathbf{D}_{{a}}(t) \mathbf{V}\bar{\mathbf{E}}_{o}^\top \infMotionAtt,
\end{align}
where $\mathbf{U} \bar{\mathbf{E}}^\top, \mathbf{V}\bar{\mathbf{E}}_{o}^\top \in \mathbb{R}^{3m \times 3n}$ link the bearing measurements variations to the agents variations.
% By applying the product rule, this yields a partition of $\mathbf{B}^+_\mathcal{G}(\chi(t))$ into two blocks as in~\eqref{eq:General_Bearing_Matrix} and we proceed in a constructive way, observing that, for any $\{c_i^t, c_i^r\}_{i=1}^n$, the matrices in~\eqref{eq:General_Bearing_Matrix} are well-defined since all the involved 
% quantities  are well established.
%
In detail, the first addendum in~\eqref{eq:b_plus_derivative_2} includes
% The block $\mathbf{D}_p(t)\mathbf{U} \bar{\mathbf{E}}^\top$ {is} 
the product of three matrices: the matrix $\bar{\mathbf{E}}^\top$ translates the variations of the agents position into relative positions (which are scaled bearings),
the block matrix $\mathbf{U}$ embeds the variation of each bearing in the manifold $\mathbb{S}^{2}$ and, finally, the block matrix $\mathbf{D}_p(t)$ accounts for the projection of each bearing $\mathbf{b}_{ij}$ in the $i$-th agent local frame. An analogous reasoning can be carried out for the second addendum in~\eqref{eq:b_plus_derivative_2} where $\mathbf{D}_a(t)$ acts similarly to $\mathbf{D}_p(t)$ since the skew matrix of the relative position vector is an orthogonal matrix.}
\end{proof}

{For any framework embedded in $\bar{\mathcal{D}}\subseteq SE(3)^n$, Proposition~\ref{prop:rigidity_matrix}, thus, provides a construction method and a general structure for the bearing rigidity matrix that can be decomposed into a part related to the formation configuration (the matrices $\mathbf{D}_p(t)$ and $\mathbf{D}_{{a}}(t)$), and a part related to the graph (the matrices $\incidenceExp$ and  $\incidenceOutExp$).  While not the focus of this work, this structure may also assist in understanding the combinatorial interpretation of certain rigidity properties. We then observe that the structure of $\mathbf{U}$ and $\mathbf{V}$ induces a level of sparsity in the extended bearing matrix, which may result in the presence of some  null columns and null rows.} 

\new{The embedding of $\bar{\mathcal{D}}$ in $SE(3)^n$ proposed in this section implies two facts. On one hand, both the infinitesimal and the trivial variations sets can be lifted from %a vector space of the dimension of
$\bar{\mathcal{I}}$ into $\mathbb{R}^{6n}$, thus considering the corresponding sets $\mathcal{S}_i^+$ and $\mathcal{S}_t^+$ such that $\vert \mathcal{S}_i^+ \vert = \vert \mathcal{S}_i \vert=q_i$ and  $ \vert \mathcal{S}_t^+ \vert = \vert \mathcal{S}_t \vert=q_t$.}
\new{These new sets are related to the null space of the matrices $\mathbf{B}^+_\mathcal{G}(\chi(t))$ and $\mathbf{B}^+_\mathcal{K}(\chi(t))$. Indeed, we first observe that $\mathcal{S}_i^+ \subset \ker{(\mathbf{B}^+_\mathcal{G}(\chi(t)))}$ and  $\mathcal{S}_t^+ \subset \ker{(\mathbf{B}^+_\mathcal{K}(\chi(t)))} $ since the extended bearing rigidity matrix can differ from the rigidity matrix because of the emergence of zero columns in correspondence of  zero entries of the vector $\infMotion^+$. This fact also justify the presence of other vectors in the null space of both $\mathbf{B}^+_\mathcal{G}(\chi(t))$ and $\mathbf{B}^+_\mathcal{K}(\chi(t))$: these are characterized by zero and non-zero entries in correspondence to non-zero and zero entries of $\infMotion^+$. In addition, because   $\mathbf{B}^+_\mathcal{K}(\chi(t))$ includes additional rows as compared to $\mathbf{B}^+_\mathcal{G}(\chi(t))$ while the structure of the two matrices is the same in terms of zero columns, we conclude that $\ker{(\mathbf{B}^+_\mathcal{G}(\chi(t)))} = \mathcal{S}_i^+ \oplus \mathcal{S}_v^+$ and $\ker{(\mathbf{B}^+_\mathcal{K}(\chi(t)))} = \mathcal{S}_t^+ \oplus \mathcal{S}_v^+$.
% Vettori del kernel di B liftati sono anche nel kernel di B^plus. Il kernel di B^plus possiede ulteriori vettori con zeri e non zeri in correispondenza delle entries liftate. this is valid fr G. For K si aggiungono righe per cui il ragionamento rimane valido sostituendo i movimenti infinitesi con itriviali e identicamente uguale per la seconda parte. tutti i set sono disgiunti per costruzione per cui somma diretta
% These new sets are related to the null space of the matrices $\mathbf{B}^+_\mathcal{G}(\chi(t))$ and $\mathbf{B}^+_\mathcal{K}(\chi(t))$ through a direct sum, i.e., $\ker{(\mathbf{B}^+_\mathcal{G}(\chi(t)))} = \mathcal{S}_i^+ \oplus \mathcal{S}_v^+$ and $\ker{(\mathbf{B}^+_\mathcal{K}(\chi(t)))} = \mathcal{S}_t^+ \oplus \mathcal{S}_v^+$, where 
The set $\mathcal{S}_v^+$, whose cardinality $q_v = \vert \mathcal{S}_v^+ \vert$ corresponds to the number of null columns in $\mathbf{B}^+_\mathcal{G}(\chi(t))$, represents the set of the \emph{virtual variations} of the  evaluated formation and it includes the command inputs inducing variations of the agent configuration that do not affect the bearing measurements but that are also not allowed by the physical constraints on the agents dynamics. On the other hand, 
% {\cut{the measurements domain $\bar{\mathcal{M}}$ is lifted into $\mathbb{S}^{2m}$: this may induce the emergence of $q_m$ null rows in correspondence to the possibly subset of the planar bearings} 
since also the measurements domain $\bar{\mathcal{M}}$ is lifted into $\mathbb{S}^{2m}$, $q_m$ null rows may characterize $\mathbf{B}_\mathcal{G}^+ (\chi(t))$ in correspondence to the the zero entries added to the vector  $\mathbf{b}_\mathcal{G}(\chi(t))$ in order to derive $\mathbf{b}_\mathcal{G}(\chi(t))$}. 
%
% In the light of these observations, the matrix $\mathbf{B}_\mathcal{G}(\chi(t))$ related to a framework in $\bar{\mathcal{D}}\subset SE(3)^n$ can be derived from the corresponding  $\mathbf{B}^+_\mathcal{G}(\chi(t))$ by removing the null rows and columns.

We finally provide a rank condition on the extended bearing rigidity matrix  that guarantees the infinitesimal rigidity of the corresponding framework embedded in any  $\bar{\mathcal{D}}$.

{
\begin{thm}%[Condition for IBR]
\label{thm:unified_condition_IBR}
A {non-colinear} $n$-agent formation modeled as a framework $(\mathcal{G}, \chi(t))$ in an arbitrary differential manifold $\bar{\mathcal{D}}$ is IBR if and only if
$\rank{\mathbf{B}^+_\mathcal{G}(\chi(t))} = \rank{\mathbf{B}^+_\mathcal{K}(\chi(t))}$.
\end{thm}
\begin{proof}
Because $ \rank{\mathbf{B}_{\mathcal G}^+} = 6n - q_v - q_i$ and $\rank{\mathbf{B}_{\mathcal K}^+} = 6n - q_v - q_t$, it holds that $\rank{\mathbf{B}^+_\mathcal{G}(\chi(t))} = \rank{\mathbf{B}^+_\mathcal{K}(\chi(t))}$ if and only if $q_i=q_t$. Due to Theorem~\ref{MainDef:thm:Inclusion_KerB_K_KerB_G}, the last equivalence is guaranteed if and only if $	\Ker{\mathbf{B}_{\mathcal{G}}(\chi(t))} = \Ker{\mathbf{B}_{\mathcal{K}}(\chi(t))}$, i.e., the framework is IBR.
\end{proof}
}
{

 \begin{table*}[t!]
 \centering
 \renewcommand{\arraystretch}{1.1}
 \resizebox{.85\textwidth}{!}{
 \begin{tabular}{|c|c|c|c|c|cc|c|}
 \cline{1-8}
 \multicolumn{4}{|c|}{$i$-th agent properties} & \multicolumn{4}{c|}{$n$-agent formation properties}\\
\cline{1-8}
$\mathcal{D}$ & $\chi_i$ & $\mathcal{M}$ & $\mathbf{b}_{ij}$ &  $\bar{\mathcal{I}}$ & \multicolumn{2}{c|}{$\infMotion$} & IBR condition \\
\hline \hline 
\multirow{2}{*}{$\mathbb{R}^{d} \quad d \in \{ 2, 3\}$} & $\mathbf{p}_i \in \mathbb{R}^{d}$  & \multirow{2}{*}{$\mathbb{S}^{d-1}$} & \multirow{2}{*}{$\vectpNorm{i}{j}$} & \multirow{2}{*}{$\mathbb{R}^{dn}$} &  \multicolumn{2}{c|}{\multirow{2}{*}{$\infMotionPos = \colvec{
\dot{\mathbf{p}}_1^{\transposed} & \dots  &\dot{\mathbf{p}}_n^{\transposed}
}^{\transposed}$}} & \multicolumn{1}{c|}{\multirow{2}{*}{$\rank{\bearingMatrixK} = dn-d-1$}} \\
& ($d$ tdofs + 0 rdofs) &  &  & & & & \\
\hline \hline
\multirow{2}{*}{$\mathbb{R}^{2}\times \mathbb{S}^1$} & $\mathbf{p}_i \in \mathbb{R}^{2}, {\theta_i} \in \left[0,2 \pi \right)$ & \multirow{2}{*}{$\mathbb{S}^1$} & \multirow{2}{*}{$\mathbf{R}_{i}^{\transposed} \vectpNorm{i}{j}$} &  \multirow{2}{*}{$\mathbb{R}^{3n}$} & \multirow{2}{*}{$\infMotion = \colvec{
\infMotionPos^{\transposed} & \infMotionAtt
}^{\transposed}$} & $
\infMotionPos = \colvec{
\dot{\mathbf{p}}_1^{\transposed} \  \dots  \ \dot{\mathbf{p}}_n^{\transposed} 
}^{\transposed}$ & \multirow{2}{*}{$\rank{\bearingMatrixK} = 3n-4$} \\
&  (2 tdofs + 1 rdofs) &  & & & & $
\infMotionAtt = \colvec{
{\dot{\theta}_1}^{\transposed} \  \dots  \ {\dot{\theta}_n}^{\transposed} 
}^{\transposed}$ &  \\
\hline
\multirow{2}{*}{$\mathbb{R}^{3}\times \mathbb{S}^1$} &  $\mathbf{p}_i \in \mathbb{R}^{3}, {\theta_i} \in \left[0,2 \pi \right)$ & \multirow{2}{*}{$\mathbb{S}^2$} & \multirow{2}{*}{$\mathbf{R}_{i}^{\transposed} \vectpNorm{i}{j}$} &  \multirow{2}{*}{$\mathbb{R}^{4n}$} & \multirow{2}{*}{$\infMotion = \colvec{
\infMotionPos^{\transposed} & \infMotionAtt
}^{\transposed}$} & $
\infMotionPos = \colvec{
\dot{\mathbf{p}}_1^{\transposed} \  \dots  \ \dot{\mathbf{p}}_n^{\transposed} 
}^{\transposed}$ & \multirow{2}{*}{$\rank{\bearingMatrixK} = 4n-5$} \\
& (3 tdofs + 1 rdofs) &  & & & & $
\infMotionAtt = \colvec{
{\dot{\theta}_1}^{\transposed} \  \dots  \ {\dot{\theta}_n}^{\transposed} 
}^{\transposed}$ &  \\
\hline \hline

\multirow{2}{*}{$\mathbb{R}^{3}\times SO(3)$} & $\mathbf{p}_i \in \mathbb{R}^{3}, \mathbf{R}_i \in SO(3)$ & \multirow{2}{*}{$\mathbb{S}^2$} & \multirow{2}{*}{$\mathbf{R}_{i}^{\transposed} \vectpNorm{i}{j}$} &  \multirow{2}{*}{$\mathbb{R}^{6n}$} & \multirow{2}{*}{$\infMotion = \colvec{
\infMotionPos^{\transposed} & \infMotionAtt
}^{\transposed}$} & $
\infMotionPos = \colvec{
\dot{\mathbf{p}}_1^{\transposed} \  \dots  \ \dot{\mathbf{p}}_n^{\transposed} 
}^{\transposed}$ & \multirow{2}{*}{$\rank{\bearingMatrixK} = 6n-7$} \\
& (3 tdofs + 3 rdofs) &  & & &  & $\infMotionAtt = \colvec{
\boldsymbol{\omega}_1^{\transposed} \  \dots  \ \boldsymbol{\omega}_n^{\transposed} 
}^{\transposed}$ &  \\
\hline
\end{tabular}
}
\caption{Summary of the principal notions related to bearing rigidity theory for the {differential manifolds} considered in {{Section}~\ref{sec:realization}}.}
\label{tab:2}
\vspace{-0.2cm}
\end{table*}

%%%%%%%%%%%%%%%%%%%%%%%%%%%%%%%%%%%%%%%%%%%%%%%%%%%%%%%%%%%%%%%%%%%%%%
\section{{Insight into Homogeneous Formations}}
\label{sec:homog_formations}

%%%%%%%%%%%%%%%%%%%%%%%%%%%%%%%%%%%%%%%%%%%%%%%%%%%%%%%%%%%%%%%%%%%%%%

{In this section, we deal with homogeneous $n$-agent formations as described in Definition~\ref{MainDef:def:Homogeneous}, further assuming that all the agents are able to translate and/or rotate around the \textit{same} directions in the global frame $\mathscr{F}_W$}.} In this case, we have that $c_i^t = c^t$ and $c_i^r=c^r$, hence $c_i = c, \forall i \in \{1 \ldots n\}$; {furthermore, $c_{tot}= \sum_{i=1}^n c_i = cn$ represents the total dofs of the formation corresponding to the dimension of the variation domain $\bar{\mathcal{I}}$.}
We first focus on bearing rigidity theory for frameworks embedded in different {differential manifolds} according to the existing literature and we show that expression~\eqref{eq:General_Bearing_Matrix} fits for all the considered cases ({Section}~\ref{sec:realization}). Along this line, {Table}~\ref{tab:generic_rigidity_matrix} specifies the variables introduced in Section~\ref{sec:unified_view} for the considered manifolds.
 Then, we take into account the IBR condition of Theorem~\ref{thm:unified_condition_IBR} and  we discuss the relations between the properties of bearing rigidity, global bearing rigidity and infinitesimal bearing rigidity %for the homogeneous formations
 ({Section}~\ref{sec:properties_homog_form}).

%%%%%%%%%%%%%%%%%%%%%%%%%%%%%%%%%%%%%%%%%%%%%%%%%%%%%%%%%%%%%%%%%%%%%%
\subsection{{Manifolds Realizations}}
\label{sec:realization}

%%%%%%%%%%%%%%%%%%%%%%%%%%%%%%%%%%%%%%%%%%%%%%%%%%%%%%%%%%%%%%%%%%%%%%

{
We recast the results from the literature about bearing rigidity for various domains with the common notation proposed in {Section}~\ref{sec:MainDef}.  {Table}~\ref{tab:2} summarizes the results presented here, anticipating also the IBR condition discussed in the sequel. Time dependency is dropped out for easing the readability.
}

\subsubsection{{Bearing Rigidity Theory in $\mathbb{R}^d$}}
\label{sec:Rd}
when ${{\mathcal{D}}=\mathbb{R}^{d}}$, ${d\in \{2,3\}}$, the attention is focused on
formations of $n$ agents 
wherein each element is modeled as a {particle}
and its configuration coincides with its  position $\mathbf{p}_i \in \mathbb{R}^d$, $i  \in \{1\ldots n\}$, in the {global} frame $\mathscr{F}_W$ that is assumed to be known by all the agents. Each element of the group is, thus, characterized by $c^t=d \in \{2,3\}$ tdofs and $c^r=0$ rdofs. 
Such frameworks, studied in \cite{2016-ZhaoZelazo-RigidityRd}, represent a suitable model, for example, for teams of mobile sensors interacting in a certain (two-dimensional or three-dimensional) area of interest.  Here, the framework model $(\mathcal{G}, \config)$, when $\bar{\mathcal{D}}=\mathbb{R}^{dn}$, the formation configuration $\config$ is associated to the \textit{positions vector} $\vectPos= \colvec{ \mathbf{p}_1^\top  \dots  \mathbf{p}_n^\top }^\top\in\mathbb{R}^{dn}$, and 
the graph $\mathcal{G}$ is \textit{undirected} since the {particle} choice allows to assume bidirectional agent iterations, meaning that neighboring agents are able to reciprocally recover bearing measurements. 
 Following the results from \cite{2016-ZhaoZelazo-RigidityRd}, the bearing rigidity matrix can be expressed as 
\begin{equation}
\label{Rd:eq:BR_Matrix}
\bearingMatrixG= \diag{s_{ij}\mathbf{P}\left(\vectpNorm{i}{j} \right)} \bar{\mathbf{E}}^{\transposed} \in \mathbb{R}^{dm \times dn}.
\end{equation}
Furthermore, given a (non-colinear) $n$-agent formation modeled as a framework $(\mathcal{G}, \config)$ in $\mathbb{R}^{dn}$, it is possible to prove (see {Lemma}~4 in~\cite{2016-ZhaoZelazo-RigidityRd} and Theorem~\ref{Rd:thm:Ker_B_K_for_n>=3} in Appendix) that its trivial variation set coincides with the $(d+1)$-dimensional set
\begin{equation}
\label{Rd:eq:Trivial_Space}
\trivialSpace= \Span{ \mathbf{1}_{n} \otimes \mathbf{I}_d,  \vectPos },
\end{equation}
describing the translation and uniform scaling of $\chi$.

Accounting for Section~\ref{sec:unified_view}, we observe that 
for ${\bar{\mathcal{D}}=\mathbb{R}^{dn}}$, independently on ${d\in \{2,3\}}$,  each agent in the group can be modeled as a rigid body having fixed attitude. From a mathematical perspective, this means that $\mathbf{R}_i = \mathbf{I}_3, \forall i \in \{1 \ldots n\}$, and that $\mathbf{V}_{ij} = \mathbf{0}_{3 \times 3}$ for every bearing measurement both for $d=2$ and $d=3$. On the other hand, in correspondence to all the edges of $\mathcal{G}$, it trivially holds that $\mathbf{U}_{ij} = \colvec{\mathbf{e}_1 \; \mathbf{e}_2 \; \mathbf{0}_3}$ when $d=2$ (embedding $\mathbb{R}^2$ into $\mathbb{R}^3$) and $\mathbf{U}_{ij}=\mathbf{I}_3$ when $d=3$. Hence, $\mathbf{D}_{a}$ in~\eqref{eq:General_Bearing_Matrix} results to be a null matrix, while $\mathbf{D}_{p}$ corresponds to the bearing rigidity matrix~\eqref{Rd:eq:BR_Matrix}.
Analysing then the null space of $\mathbf{B}^+_\mathcal{G}(\chi)$ as in~\eqref{eq:General_Bearing_Matrix}, the virtual trivial variations result to be $q_v = 4n$ when $d=2$ and $q_v = 3n$ when $d=3$ and these corresponds to the three rotational movements, in addition to the translation along the $z$-axis of the 3D global frame for the $d=2$ case.

%%%%%%%%%%%%%%%%%%%%%%%%%%%%%%%%%
\subsubsection{{Bearing Rigidity Theory in $\mathbb{R}^d \times \mathbb{S}^1$}}
\label{sec:RdxSO(2)}
 in a formation wherein each agent configuration is defined in $\mathbb{R}^d \times \mathbb{S}^1$, $d \in \{ 2, 3\}$, the $n$ components are all characterized by $c^t=d$ tdofs and $c^r=1$ rdof controllable in a decoupled way.
This is, for instance, the case of teams of unicycle-modeled ground robots ($d\!=\!2$) or of standard under-actuated quadrotors~($d\!=\!3$) whose controllable variables are the position and the yaw~angle.
The described formation can be modeled as a framework $(\mathcal{G}, \config)$ in $\bar{\mathcal{D}} = \left( \mathbb{R}^{d} \times \mathbb{S}^1 \right)^n$. In this case, $\config$ is associated~to  both the \textit{positions vector} $\mathbf{p} = \colvec{  \mathbf{p}_1^\top  \dots  \mathbf{p}_n^\top }^\top \in \mathbb{R}^{dn}$ and~the \emph{attitudes vector} ${\boldsymbol{\theta}} =\colvec{  {\theta_1}  \dots  {\theta_n} }^\top \in \left[ 0, 2\pi \right)^n$, while 
the graph~$\graphG$ is
\emph{directed}, since we assume that agents do not have access to the {global} frame. 
Basing on~\cite{2014-ZelazoFranchiGiordano-RigiditySE2, 2015-ZelazoGiordanoFranchi-RigiditySE2}, the  bearing rigidity matrix can be written as
\begin{equation}
\label{RdxSO2:eq:Rigidity_Matrix}
\bearingMatrixG =
 \colvec{
\check{\mathbf{D}}_p \incidenceExp^{\transposed} &  \check{\mathbf{D}}_a \incidenceOut^{\transposed}  
} \in \mathbb{}^{dm \times \left(d+1\right)n},
\end{equation}
where  $\incidenceExp\in\mathbb{R}^{dn\times dm}$,  $\incidenceOut\in\mathbb{R}^{n\times m}$ are derived from $\graphG$ and
\begin{align}
\check{\mathbf{D}}_p &= \diag{s_{ij}\mathbf{R}_i^{\transposed} \mathbf{P}\left( \vectpNorm{i}{j} \right)}\in \mathbb{R}^{dm \times dm} \;\; \text{if}\; d \in \{ 2, 3\}, \nonumber\\
\check{\mathbf{D}}_a &= \begin{cases}
-\diag{\mathbf{R}_{i}^{\transposed}\vectpNorm{i}{j}^\perp}\in \mathbb{R}^{dm \times m}  & \text{if} \; d=2 \\
-\diag{\mathbf{R}_{i}^{\transposed} \skewMatrix{\vectpNorm{i}{j}}\mathbf{v}}\in \mathbb{R}^{dm \times m}  & \text{if} \; d=3
\label{eq:blocks}
\end{cases}
\end{align}
In~\eqref{eq:blocks}, $\vectpNorm{i}{j}^\perp = \mathbf{R}\left(\frac{\pi}{2}\right) \vectpNorm{i}{j} \in \mathbb{R}^2$, $\mathbf{R}\left(\frac{\pi}{2}\right) \in SO(2)$ is the (unit) vector perpendicular to $ \vectpNorm{i}{j} $ on the plane, while $\mathbf{R}_{i}\in SO(d)$ identifies the orientation of $\mathscr{F}_i$ with respect to $\mathscr{F}_W$. 
For a formation on a plane ($d=2$), 
the orientation of each agent is (completely) specified by an angle that is univocally associated to a rotation matrix $\mathbf{R}_{i}= \mathbf{R}\left({\theta_i} \right) \in SO(2)$, whereas for the 3D case ($d=3$), $\mathbf{R}_i =\mathbf{R} \left( {\theta_i}, \mathbf{v} \right) \in SO(3)$ denotes the rotation of angle ${\theta_i}\in \left[ 0, 2\pi \right)$ around the (unit) vector $\mathbf{v}=\textstyle{\sum_{h=1}^3} v_h \mathbf{e}_h$ with $v_h \in \mathbb{R}$, identifying a fixed direction in $\worldframe$. 

The trivial variations coincide with the translation and uniform scaling of the entire configuration, jointly 
with the \textit{coordinated rotation}, namely the equal rotation of all the agents jointly with the equal rotation of the whole formation around its center.  The {coordinated rotation subspace $\coordRot$} is formally determined as 
\begin{equation}
\coordRot = \begin{cases}
\Span{\colvec{
	\left(  \identityMatrix{n} \otimes \mathbf{R}\left(\pi/2\right)\right) \vectPos\\
	\vectOnes{n}
	}}, & \: \text{if} \  d=2 \\
\Span{\colvec{
	\left(  \identityMatrix{n} \otimes \skewMatrix{\mathbf{v}} \right) \vectPos \\
	\vectOnes{n}
	}}, & \: \text{if} \ d=3
\end{cases}
\end{equation}
and, since $\Dim{\coordRot}=1$ for $d \in\{2,3\}$, the trivial variations set has dimension $q_t=d + 2$. This formally results 
\begin{equation}
	\label{RdxSO2:eq:Trivial_Space}
	\trivialSpace = \Span{ \colvec{
	\mathbf{1}_{n} \otimes \mathbf{I}_d\\
	\mathbf{0}_n\\
	}, \colvec{
	\vectPos \\
	\mathbf{0}_n
	}, \mathcal{R}_{\circlearrowleft}}.
\end{equation}

Adopting the unified view of Section~\ref{sec:unified_view}, it is convenient to distinguish between the differential manifolds  $\mathbb{R}^2 \times \mathbb{S}^1 = SE(2)$ and $\mathbb{R}^3 \times \mathbb{S}^1$. Indeed, in the first case ($d=2$), assuming that the formation evolves on the $(xy)$-plane of $\mathscr{F}_W$ identified by $\mathbf{e}_1 \times \mathbf{e}_2$, we have that $\mathbf{U}_{ij} = \colvec{\mathbf{e}_1 \; \mathbf{e}_2 \; \mathbf{0}_3}$ for all the measurements, and that
{$\mathbf{R}_i = \mathbf{R}({\theta_i}, \mathbf{e}_{3})$ $\forall i \in \{ 1 \ldots n\}$}, according to the axis-angle representation $\mathbf{R}(\cdot, \cdot)$ of a 3D rotation, hence, $\mathbf{V}_{ij} = \colvec{\mathbf{0}_{3\times 2} \;\mathbf{e}_3}$ in correspondence to any edges of $\mathcal{G}$. In the second case ($d=3$), instead, $\mathbf{U}_{ij} = \mathbf{I}_3$, while $\mathbf{R}_i = \mathbf{R}({\theta_i}, \mathbf{v})$, 
and, thus, $\mathbf{V}_{ij} = \colvec{\mathbf{0}_{3\times 2} \;\mathbf{v}}$. 
{In particular, one can check that  $\colvec{\mathbf{0}_{3\times 2} \; [(\vectpNorm{i}{j}^\perp)^\top \; 0]^\top} = \skewMatrix{[\vectpNorm{i}{j}^\top \; 0]^\top}\colvec{\mathbf{0}_{3\times 2} \; \mathbf{e}_3}$ and $\colvec{\mathbf{0}_{3\times 2} \; \skewMatrix{\vectpNorm{i}{j}}\mathbf{v}}=\skewMatrix{\vectpNorm{i}{j}}\colvec{\mathbf{0}_{3\times 2} \; \mathbf{v}}$.} Thus, we conclude that the expression~\eqref{eq:General_Bearing_Matrix} can be reduced to the bearing rigidity matrix~\eqref{RdxSO2:eq:Rigidity_Matrix}.
Furthermore, when $\bar{\mathcal{D}}=(\mathbb{R}^d \times \mathbb{S}^1)^n$, the set $\mathcal{S}_v$ includes the $2n$ unfeasible rotational movements of the agents, in addition to the translation along the $z$-axis of $\mathscr{F}_W$ when $d=2$. Summing up, concerning the virtual trivial variations, we have that $q_v = 3n$ when $d=2$ and $q_v = 2n$ when $d=3$.

%%%%%%%%%%%%%%%%%%%%%%%%%%%%%%%%%%%%%%%%%%%%%%%%%%%%%%%%%%%%%%%%%%%%%%

%%%%%%%%%%%%%%%%%%%%%%%%%%%%%%%%%%%%%%%%%%%%%%%%%%%%%%%%%%%%%%%%%%%%%%
\subsubsection{{Rigidity Theory in $SE(3)$}}
\label{sec:SE3}

the recent work~\cite{2016-MichielettoCenedeseFranchi-RigiditySE3} accounts for formations of agents  equipped with a bearing sensor whose configuration is given by an element of $SE(3)$. An example is given by a swarm of fully-actuated aerial platforms provided with on-board omnidirectional cameras.
The considered group of agents can be modeled as a framework $\framework{\config}$ in $\bar{\mathcal{D}} =SE(3)^n$, where $\graphG$ is a \textit{directed} graph and $\config$  deals with the \textit{positions vector} $\mathbf{p} = \colvec{ \mathbf{p}_1^\top  \dots  \mathbf{p}_n^\top }^\top \in \mathbb{R}^{3n}$, and the $(3n \times 3)$~ \textit{attitudes matrix} $\mathbf{R}_a = \colvec{ \mathbf{R}_1^\top \dots \mathbf{R}_n^\top }^\top\in SO(3)^{n}$, stacking all the agent position vectors and rotation matrices, respectively. 
Note that the commands space $\mathcal{I}$ of each $i$-th agent includes its linear velocity $\dot{\mathbf{p}}_i \in \mathbb{R}^3$ and its angular velocity $\boldsymbol{\omega}_i \in \mathbb{R}^3$, both expressed in $\worldframe$.
Given these premises, it is thus possible to prove that the bearing rigidity matrix, belonging to $ \mathbb{R}^{3m\times 6n}$, turns out to be
\begin{equation}
\label{SE3:eq:Bearing_Matrix}
	\bearingMatrixG \! = \! \colvec{
	\diag{s_{ij} \mathbf{R}_{i}^{\transposed} \mathbf{P}\left(\vectpNorm{i}{j} \right)}\bar{\mathbf{E}}^{\transposed} \!\! & \! -\diag{ \mathbf{R}_i^{\transposed}\skewMatrix{\vectpNorm{i}{j}} } \bar{\mathbf{E}}^{\transposed}_{o} 
	}.
\end{equation} 

Comparing~\eqref{SE3:eq:Bearing_Matrix} with~\eqref{RdxSO2:eq:Rigidity_Matrix}, we observe that 
the translation, uniform scaling and coordinated rotation are trivial variations also for a framework $\framework{\config}$ in $SE(3)^n$, however the concept of coordinated rotation has to be redefined since the agents orientation is no  longer controllable only via a single angle.\\
Specifically, it has been proven in~\cite{2016-MichielettoCenedeseFranchi-RigiditySE3}
\begin{equation}
\label{SE3:eq:CoordRot_Space}
\coordRot = \Span{\left\{\colvec{
\left(  \identityMatrix{n} \otimes \skewMatrix{\mathbf{e}_{h}} \right) \vectPos \\
\vectOnes{n}\otimes \mathbf{e}_h \\
}\right\}_{h=1,2,3} }.
\end{equation}
Therefore, the following set $\trivialSpace$ tha  $\Dim{\trivialSpace}=7$,
\begin{equation}
\label{SE3:eq:Trivial_Space}
\trivialSpace = \Span{ \colvec{
	\mathbf{1}_{n} \otimes \mathbf{I}_3\\
	\mathbf{0}_{3n}\\
	}, \colvec{
	\vectPos \\
	\mathbf{0}_{3n}
	}, \mathcal{R}_{\circlearrowleft}}.
\end{equation}

When $\bar{\mathcal{D}}=SE(3)^n$, each agent is characterized by $c=6$ dofs and, thus, it can vary its position and attitude in any direction of the 3D space.
Therefore, it follows that $\mathbf{U}_{ij}=\mathbf{V}_{ij} = \mathbf{I}_3$ and the matrix $\mathbf{B}^+_{\mathcal{G}}(\chi)$ corresponds to that in~\eqref{SE3:eq:Bearing_Matrix}.
Clearly, when the agents act in $SE(3)$, no embedding in an higher dimensional manifold is needed, which translates into the absence of virtual trivial motions, and, correspondingly, no null columns appear in $\mathbf{B}^+_{\mathcal{G}}(\chi)$.

%%%%%%%%%%%%%%%%%%%%%%%%%%%%%%%%%%%%%%%%%%%%%%%%%%%%%%%%%%%%%%%%%%%%%%
\subsection{{Rigidity Properties}}
\label{sec:properties_homog_form}

% For homogeneous formations the next result holds. Its validity is guaranteed by the fact that $\rank{{\mathbf{B}_\mathcal{G}(\chi(t))}}= \rank{{\mathbf{B}^+_\mathcal{G}(\chi(t))}}$ and in the light of {Theorem}~4 in~\cite{2016-ZhaoZelazo-RigidityRd}, {Theorem}~III.6 in~\cite{2014-ZelazoFranchiGiordano-RigiditySE2} and {Theorem}~3 in~\cite{2016-MichielettoCenedeseFranchi-RigiditySE3} which provide a necessary and sufficient condition on the rank of the bearing rigidity matrix that can be used to check whether given formation is IBR   in $\mathbb{R}^{dn}$, $(\mathbb{R}^d \times \mathbb{S}^1)^n$, $SE(3)^n$, respectively.

{Accounting for Theorem~\ref{thm:unified_condition_IBR}, the next result provides a necessary and sufficient condition to check if a given homogeneous formation is IBR. Its validity is guaranteed by observing that $\rank{{\mathbf{B}_\mathcal{G}(\chi(t))}}= \rank{{\mathbf{B}^+_\mathcal{G}(\chi(t))}}$ together with the results of {Theorem}~4 in~\cite{2016-ZhaoZelazo-RigidityRd}, {Theorem}~III.6 in~\cite{2014-ZelazoFranchiGiordano-RigiditySE2} and {Theorem}~3 in~\cite{2016-MichielettoCenedeseFranchi-RigiditySE3} dealing with homogeneous frameworks embedded in $\mathbb{R}^{dn}$, $(\mathbb{R}^d \times \mathbb{S}^1)^n$, $SE(3)^n$, respectively.}

\begin{thm} %[Condition for IBR]
\label{thm:Condition_IBR}
{A homogeneous} framework $(\mathcal{G}, \chi(t))$ in $\bar{\mathcal{D}}$, where each agent has $c\leq 6$ dofs is IBR if and only if $\rank{{\mathbf{B}^+_\mathcal{G}(\chi(t))}}=cn-c-1$.
\end{thm}

\begin{proof}
\new{Note that $\rank{\mathbf{B}^+_\mathcal{K}(\chi(t))} = 6n-q_t-q_v$ where the difference $6n-q_v$ corresponds to the total dofs  of the formations, namely to $c_{tot} = cn \leq 6n$ in homogeneous cases.
%For {a homogeneous} formation, the difference $6n-q_v$ corresponds to the total dofs $c_{tot} = cn \leq 6n$, of the formations due to the fact that the set of virtual variations identify the agents physical constraints. 
On the other hand, from~\eqref{Rd:eq:Trivial_Space},~\eqref{RdxSO2:eq:Trivial_Space} and~\eqref{SE3:eq:Trivial_Space}, it can be inferred that for {a homogeneous} framework  it holds that $q_t=c+1$, i.e., shape preservation is ensured when the formation acts as a unique rigid body having $c$ DoFs and  when it scales. Theorem~\ref{thm:unified_condition_IBR}, thus, ensures that a non-colinear homogeneous formation is IBR if and only if $\rank{\mathbf{B}^+_\mathcal{G}(\chi(t))} = \rank{\mathbf{B}^+_\mathcal{K}(\chi(t))} = (6n-q_v)-q_t = cn - (c+1)$
}
\end{proof}

%%%%%%%%

 As already highlighted in Section~\ref{sec:framework}, the \new{IBR property} is generally studied for dynamic frameworks, while \new{the BR and the GBR ones are} usually discussed for static frameworks. Nonetheless, these last two properties can be also stated for dynamic frameworks over time, and in particular they can hold for any $t$. 
 %
% Given these premises, we now show that, although a GBR framework is also BR independent of $\bar{\mathcal{D}}$ (Proposition~\ref{prop:GBR_BR}) and for any time $t$, the inverse implication is not always true. That is, a BR framework is also GBR only if this is embedded in $\bar{\mathcal{D}}=\mathbb{R}^{dn}$, $d \in \{2,3\}$. 
 %
 %On the other hand, besides the differential manifold $\bar{\mathcal{D}}$, an IBR framework is also GBR over time, and {vice versa}. 
 %This is proved in the next theorem, whose results are valid in the proposed unified view as well as in those discussed in the literature. 
 %\new{On the other hand, besides the differential manifold $\bar{\mathcal{D}}$, an IBR framework is also GBR over time and is equivalent to a BR framework. 
 %This complete characterization is proved in the next theorem, whose results are valid in the proposed unified view as well as in those discussed in the literature.} 

%
%As already highlighted in Section~\ref{sec:framework}, the infinitesimal bearing rigidity is generally studied for dynamic frameworks, while the bearing rigidity and the global bearing rigidity are properties usually discussed for static frameworks. Nonetheless, these last two properties can be also stated for dynamic frameworks over time, and in particular they can hold for any $t$. 
Given these premises, \new{the next theorem provides a complete characterization to  clarify the relation between BR, GBR and IBR properties for non-colinear frameworks embedded in any differential manifold $\bar{\mathcal{D}}$.}

\begin{thm}
\label{thm:IBR_GBR}
\new{Given a differential manifold $\bar{\mathcal{D}}$ for  any $t$,
\begin{itemize}
\item [$i)$] a framework $(\mathcal{G},{\chi(t)})$ is IBR if and only if it is BR;
\item [$ii)$] a framework $(\mathcal{G},{\chi(t)})$ is IBR if it is GBR.
\end{itemize}}
\end{thm}
%
% \new{
%  \begin{thm}
%  \label{thm:rigidity_proprieties}
% In correspondence to any {differential configuration manifold} $\bar{\mathcal{D}}$, it holds that
% \begin{itemize}
%     \item[a)] a framework $(\mathcal{G},{\chi(t)})$ is IBR if and only if it is BR for any $t$;
%     \item[b)] a framework $(\mathcal{G},{\chi(t)})$ is IBR only if it is GBR for any $t$.
%\end{itemize}
% \end{thm}
%
\new{
\begin{proof}
$i)$ Proceeding by contrapositive, if $(\mathcal{G},{\chi(t)})$ is not BR for any $t$, there exists at least a configuration $\chi(t^\prime)$ in the neighborhood $\setNeigh{\config}$ of $\chi(t)$ such that $\chi(t^\prime)\in \setEquivalence{\config(t)}\setminus \setCongruence{\config(t)}$, and therefore the framework $(\mathcal{G},{\chi(t)})$ would result to be IBF according to the consequence of Definition~\ref{MainDef:def:IBR}.
To prove the reverse, we assume that the framework $(\mathcal{G},{\chi(t)})$ is BR  for any $t$.
%we first observe that a framework $(\mathcal{G},{\chi(t)})$ is BR  for any $t$ if $\setEquivalence{\config(t)} \cap \setNeigh{\config(t)} = \setCongruence{\config(t)} \cap \setNeigh{\config(t)}$ for any $t$ (Definition~\ref{MainDef:def:BR}). 
Ab absurdo, if the framework was not IBR (IBF) there would be a variation deforming $\chi(t)$ into $\chi(t^\prime) \in \setEquivalence{\config(t)}\setminus \setCongruence{\config(t)}$, implying a contradiction with respect to Definition~\ref{MainDef:def:BR}. \\ 
%
%Ab absurdo, suppose that $\mathcal{Q}(\chi(t)) = \mathcal{C}(\chi(t))$, $\forall t\in [0,1]$ and that there exists an analytical path $\gamma:[0,1]\rightarrow \bar{\mathcal{D}}$ s.t. $\gamma(0)=\chi$ and $\gamma(t)\in \chi$
$ii)$ It is a direct implication of Proposition~\ref{prop:GBR_BR} and i).
%
%
% Given the framework $(\mathcal{G},\chi(t))$, for any time $t$,
% the relation~\eqref{MainDef:eq:BR_matrix} implies that for any $\chi(t') \in \mathcal{Q}(\chi(t))$, it holds that $\Ker{{\mathbf{B}_\mathcal{G}(\chi(t))}}=\Ker{{\mathbf{B}_\mathcal{G}(\chi(t')}}$. Moreover, for any $\chi(t') \in \mathcal{C}(\chi(t))$, it follows that $\Ker{{\mathbf{B}_\mathcal{K}(\chi(t))}}=\Ker{{\mathbf{B}_\mathcal{K}(\chi(t')}}$.\\
% According to {Definition}~\ref{MainDef:def:GBR}, for  a GBR (dynamic) framework $(\mathcal{G},\chi(t))$ 
% it holds that $\mathcal{C}(\chi(t)) = \mathcal{Q}(\chi(t))$ for any $t$ and, because of the given premises, this implies $\Ker{{\mathbf{B}_\mathcal{G}(\chi(t))}}=\Ker{{\mathbf{B}_\mathcal{K}(\chi(t))}}$, namely $(\mathcal{G},\chi(t))$ is also IBR. Conversely, the equivalence of $\mathcal{S}_{i}$ and $\mathcal{S}_{t}$ guaranteed by infinitesimal bearing rigidity entails $\mathcal{C}(\chi(t)) = \mathcal{Q}(\chi(t))$ for any time $t$. Ab absurdo, indeed, consider $\chi(t')\in \mathcal{Q}(\chi(t))$ such that $\chi(t')\notin \mathcal{C}(\chi(t))$. Then there exists $\infMotion \in \mathcal{S}_{i}$ such that $\infMotion \notin \mathcal{S}_{t}$ but this is in contrast with IBR hypothesis.
\end{proof}
}

%A consequence of Proposition~\ref{prop:GBR_BR} and Theorem~\ref{thm:IBR_GBR} is that an IBR framework $(\mathcal{G},{\chi(t)})$ is also BR {for any time}, independently of {the configuration space} $\bar{\mathcal{D}}$.

\new{The following results can also be stated for the frameworks  embedded in $\bar{\mathcal{D}} = \mathbb{R}^{dn}$, $d \in \{2,3\}$.}

\new{
\begin{prop}
\label{prop:BR_GBR}
A BR framework $(\mathcal{G},{\chi(t)})$ in $\bar{\mathcal{D}} = \mathbb{R}^{dn}$, $d \in \{2,3\}$, is also GBR.
\end{prop}
\begin{proof}
According to {Theorem}~1 and {Theorem}~3 in\new{~\cite{2016-ZhaoZelazo-RigidityRd}}, a framework {$(\mathcal{G},{\chi(t')})$} is BE/BC to $(\mathcal{G},{\chi(t)})$ if and only if its corresponding positions vector belongs to $
\Ker{\mathbf{B}_{\mathcal{G}}({\chi(t))}}$/$
\Ker{\mathbf{B}_{\mathcal{K}}({\chi(t))}}$, moreover, for a BR framework it holds that $\Ker{\mathbf{B}_{\mathcal{G}}({\chi(t))}} \subseteq \Ker{\mathbf{B}_{\mathcal{K}}({\chi(t)})}$. These two facts imply that a framework $(\mathcal{G},{\chi(t)})$ in $\bar{\mathcal{D}} = \mathbb{R}^{dn}$, $d \in \{2,3\}$, is GBR if the condition $\Ker{\mathbf{B}_{\mathcal{G}}({\chi(t)})} = \Ker{\mathbf{B}_{\mathcal{K}}({\chi(t)})}$ holds. Hence,  because of Theorem~\ref{MainDef:thm:Inclusion_KerB_K_KerB_G}, the bearing rigidity property ensures the global rigidity property.
\end{proof}

We observe that the requirement on the null spaces equivalence for GBR property derived in the proof of Proposition~\ref{prop:BR_GBR} coincides with the definition of IBR property provided in {Definition}~\ref{MainDef:def:IBR}. Thus, a IBR framework in $\mathbb{R}^{dn}$ is also GBR for any time  and {vice versa}, leading to the next corollary.

\begin{cor}
For frameworks embedded in $\bar{\mathcal{D}} = \mathbb{R}^{dn}$, $d \in \{2,3\}$, bearing rigidity, global bearing rigidity and infinitesimal bearing rigidity are equivalent properties.
\end{cor}
}

\new{The equivalence among rigidity properties is not valid when introducing the rotational dofs as for frameworks in $(\mathbb{R}^d \times \mathbb{S}^1)^n$, $d \in \{2,3\}$, or in $SE(3)^n$. Figure~\ref{fig:counter_example} provides an example (panel (a)) of a IBR framework in $(\mathbb{R}^2 \times \mathbb{S}^1)^n$ composed of $n=4$ agents, that is not GBR since it is possible to deform into a BE and not BC framework (panel (b)).}

\begin{figure}[t]
    \centering
    \includegraphics[width =0.43\textwidth]{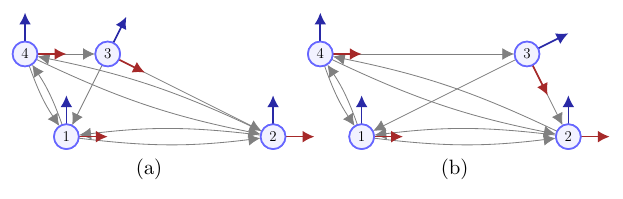}
    \caption{\new{Example of IBR and not GBR framework in  $\bar{\mathcal{D}}\!=\!(\mathbb{R}^2 \!\times\! \mathbb{S}^1)^n$ with $n=4$. Blue circles, blue-red arrows, gray arrows denote respectively the agents, the corresponding local frames and the existing bearing measurements.}}
    \label{fig:counter_example}
\end{figure}
\section{{Applications of Unified Rigidity Theory}}
\label{sec:applications}

%%%%%%

In these case studies, time dependency is omitted for brevity.

\subsection{{Homogeneous Formation in $\mathbb{R}^3 \times \mathbb{S}^2$ Case Study}}
\label{sec:new_homogeneous_rigidity}
As a first application of unified rigidity theory, we discuss a homogeneous formation case study that has not been considered in the literature and show how this can be well accommodated in the framework of Section~\ref{sec:unified_view}.

We account for the case study in Figure~\ref{fig:mixed_formation}(a), where a formation is composed of four aerial platforms (blue circle) with bearings coming from gimbal cameras, which, despite the free movement of the agents in the 3D space, keep their image planes aligned with the horizon and are denied the roll motion: for each agent {$c=5$ ($c^t=3$ and $c^r=2$)}, and $\mathcal{D}=\mathbb{R}^{3}\times \mathbb{S}^2$. 
{The considered $n=4$ agents group is modeled as {a homogeneous} agent framework $(\mathcal{K}, \chi)$, where $\chi = (\chi_1 \ldots \chi_4)$ with $\chi_1, \chi_2,\chi_3,\chi_4 \in \mathbb{R}^3 \times \mathbb{S}^2$. Note that the relative bearing measurement between each pair of agents is expressed in their local frames and the underlying graph is complete, thus the formation is IBR, according to Definition~\ref{MainDef:def:IBR}.}

According to Proposition~\ref{prop:rigidity_matrix} and the related discussion held in Section~\ref{sec:homog_formations} for the homogeneous case, all the edges can vary in the 3D space in the same way, in particular their rotations are constrained around the $y$-axis and $z$-axis of $\mathscr{F}_W$. Hence, $\mathbf{U}_{ij} = \mathbf{I}_3$ and $\mathbf{V}_{ij} = [ \mathbf{0}_3 \; \mathbf{e}_2 \; \mathbf{e}_3]$ in correspondence to any ${e}_{ij} \in \mathcal{E}$. The bearing matrix $\mathbf{B}^+_{\mathcal{K}}(\chi) \in \mathbb{R}^{36 \times 24}$ takes the form as in Figure~\ref{fig:mixed_rigidity_matrix}(a), where four null columns can be highlighted.

%\cut{Two observations are now in order. {First, for Theorem~\ref{thm:Condition_IBR}, it holds that $\rank{\mathbf{B}^+_{\mathcal{K}}(\chi)}=4(3+2)-(3+2)-1 = 14$ and, thus  $\nullity{\mathbf{B}^+_{\mathcal{K}}(\chi)}=10$.} Secondly, {recalling that the null space of  $\mathbf{B}^+_{\mathcal{K}}(\chi)$ accounts both for the trivial and the virtual variations of the framework}, we can identify {$q_t = c+1 = 6$} trivial variations, which are three translations, two coordinated rotations, and one scaling, in addition to {$q_v =4$} virtual variations that are due to the four null columns of $\mathbf{B}^+_{\mathcal{K}}(\chi)$ and correspond to the actual constraints on the agents motion. Actually, these appear in the {orthonormalized} basis of $\ker{(\mathbf{B}^+_{\mathcal{K}}(\chi))}$ as vectors having just one non-null component corresponding to the agent rotational velocity around the $x$-axis of $\mathscr{F}_W$.}
\new{From Theorem~\ref{thm:Condition_IBR}, it holds that $\rank{\mathbf{B}^+_{\mathcal{K}}(\chi)}=4(3+2)-(3+2)-1 = 14$. This fact can be numerically verified and implies that $\nullity{\mathbf{B}^+_{\mathcal{K}}(\chi)}=q_t + q_v = 10$. In particular, we can identify  {$q_t = c+1 = 6$} trivial variations, which are three translations, two coordinated rotations, and one scaling, in addition to {$q_v =4$} virtual variations that are due to the four null columns of $\mathbf{B}^+_{\mathcal{K}}(\chi)$ and correspond to the actual constraints on the agents motion. 
}

\begin{figure}[t!]
\centering
\includegraphics[width=\columnwidth]{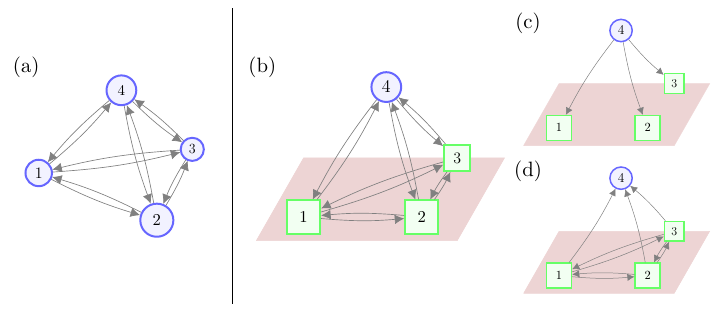}
 \caption{{Sensing graphs for the considered case studies. (a) Homogeneous formation. (b) Heterogeneous formation $\mathcal{K}=\mathcal{G}_1 \cup \mathcal{G}_2$, with (c) aerial bearing subgraph $\mathcal{G}_1$, (d) terrestrial bearing subgraph $\mathcal{G}_2$.}}
 \label{fig:mixed_formation}
\vspace{-0.0cm}
\end{figure} 
%%%%%
\begin{table}[t!]
{
    \centering
    \begin{tabular}{|c|c|c|}
    \hline
   $(i,j)$  & $\mathbf{U}_{ij}$ & $\mathbf{V}_{ij}$ \\
   \hline \hline
   $(1,2),(1,3),(2,1),(2,3),(3,1),(3,2)$      & $[\mathbf{e}_1 \; \mathbf{e}_2 \; \mathbf{0}_{3}]$ & $[\mathbf{0}_{3\times 2} \;\mathbf{e}_3]$  \\
    $(1,4), (2,4), (3,4)$  & $\mathbf{I}_3$  & $[\mathbf{0}_{3\times 2} \;\mathbf{e}_3]$  \\
    $(4,1), (4,2), (4,3) $ & $\mathbf{I}_3$  &  $\mathbf{I}_3$ \\
         \hline
    \end{tabular}
    }
    \caption{{Matrices for the heterogeneous formation case study.}}
    \label{tab:selection_matrices}
\end{table}
%%%%%

%%%%%%
\subsection{{Heterogeneous Formation Case Study}}
\label{sec:heterogeneous_rigidity}

The discussion carried out in {Section}~\ref{sec:unified_view} about the structure of the rigidity matrix turns out to be {particularly} useful for the case of heterogeneous formations made up of agents {whose configurations evolve in different differential manifolds.} 

{To validate this observation, we account for the case study in the right part of {Figure}~\ref{fig:mixed_formation} (panels (b)-(c)-(d)) referring to a formation involving three unicycle-modeled terrestrial robots (green squares), and a fully-actuated aerial platform (blue circle). We assume that the agents are able to retrieve relative bearings expressed in their local frames according to the depicted complete graph $\mathcal{K}$.} 
{Hence, the considered $n=4$ agents group {is modeled as the IBR framework} $(\mathcal{K}, \chi)$, where $\chi = (\chi_1 \ldots \chi_4)$ with $\chi_1, \chi_2,\chi_3 \in \mathbb{R}^2 \times \mathbb{S}^1$, and $\chi_4 \in SE(3)$.} 

{Assuming that the global frame is oriented so that the terrestrial vehicles can move on the $(xy)$-plane of $\mathscr{F}_W$, the matrices $\mathbf{U}$ and $\mathbf{V}$ in \eqref{eq:General_Bearing_Matrix} are determined based on Table \ref{tab:selection_matrices}. 
Hence, the rigidity matrix $\mathbf{B}^+_{\mathcal{K}}(\chi) \in \mathbb{R}^{36 \times 24}$ has the structure reported in {Figure}~\ref{fig:mixed_rigidity_matrix}(b).}

{The set  $\mathcal{S}_t$ counts $q_t=5$ trivial variations consisting of the translation of the whole formation on the $(xy)$-plane of the global frame, its coordinated rotation around the $z$-axis of $\mathscr{F}_W$, and its scaling (indeed, the third, sixth, ninth, tenth and eleventh columns of $\mathbf{B}^+_{\mathcal{K}}(\chi)$ results to be linearly dependent with respect to the remaining ones). In addition, the matrix has six null columns, meaning that the set $\mathcal{S}_v$ of virtual variations has cardinality $q_v =6$. This, in fact, coincides with the (unfeasible) rotations of the terrestrial robots around their $x$ and $y$-axis, and is related to the zero columns in  $\mathbf{B}^+_{\mathcal{K}}(\chi)$.
From these premises, it follows that $\nullity{\mathbf{B}^+_{\mathcal{K}}(\chi)}=11$ and, finally, that $\rank{\mathbf{B}^+_{\mathcal{K}}(\chi)}=13$, as can be numerically verified.}

{To conclude, we note that the given IBR formation}
results from the union of two sensing graphs, each one corresponding with the measurements obtained by either the terrestrial robots ($\mathcal{G}_1$) or the aerial platform ($\mathcal{G}_2$), as in {Figures}~\ref{fig:mixed_formation}(c) and \ref{fig:mixed_formation}(d). Focusing on the infinitesimal rigidity of the resulting two frameworks, we can observe that  any rotation of the aerial vehicle is an infinitesimal variation for $(\mathcal{G}_1,\chi)$, and similarly, any rotation of the terrestrial vehicle is an infinitesimal variation for {$(\mathcal{G}_2,\chi)$},  concluding that the two frameworks are not IBR differently  from their union $(\mathcal{K},\chi)$.   

\begin{figure}[t!]
\begin{center}
\includegraphics[width=0.95\columnwidth]{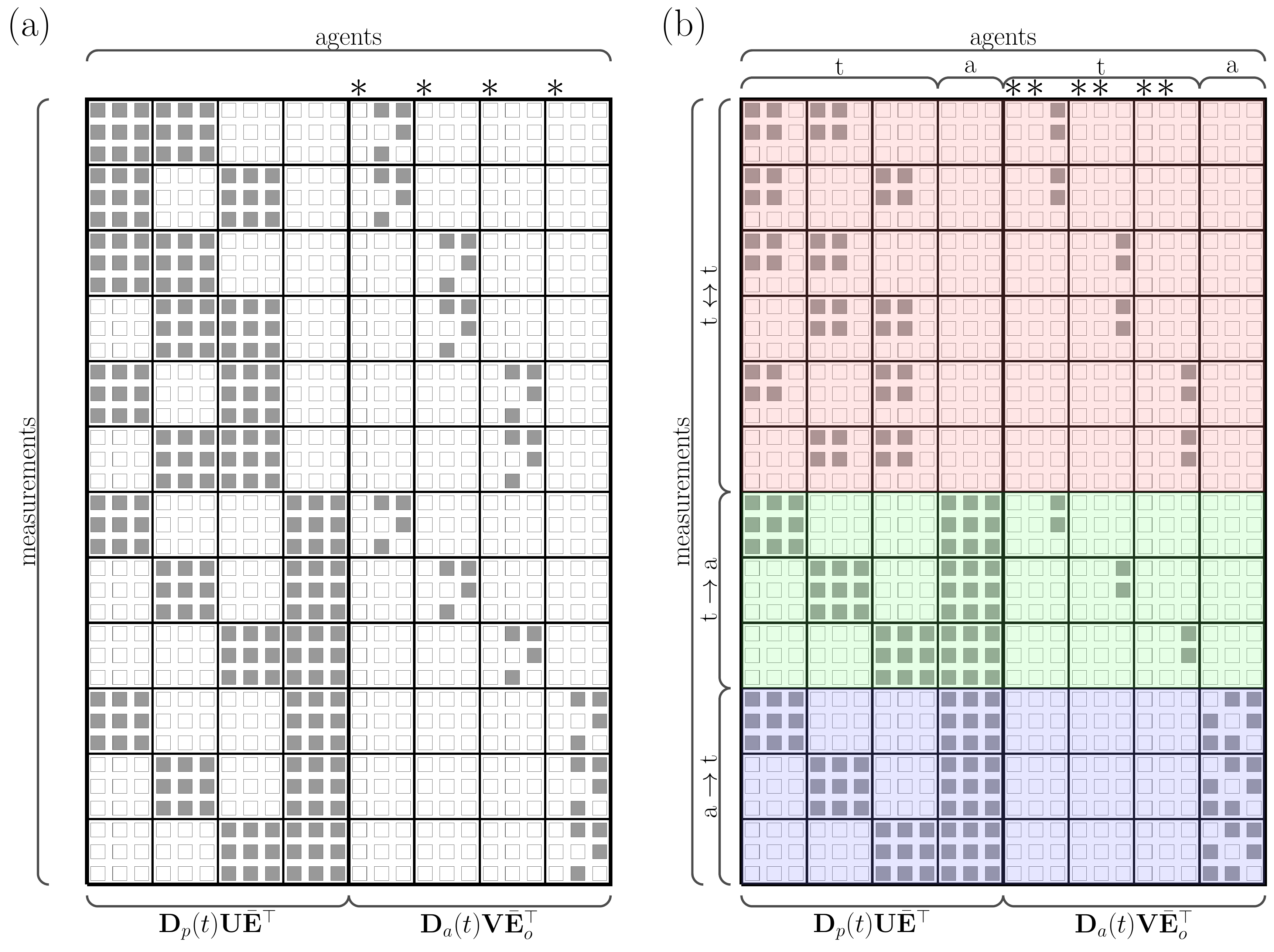}
\end{center}
\caption{{Structure of the bearing rigidity matrix for the considered case studies. (a) Homogeneous and (b) heterogeneous formations.} Dark squares indicate non-null values of the matrix; {in the heterogeneous case t and a labels refer respectively to terrestrial and aerial agents. Null columns are indicated with an asterisk.}}
\label{fig:mixed_rigidity_matrix}
\vspace{-0.2cm}
\end{figure}

%%%%%%%%%%%%%%%%%%%%%%%%%%%%%%%%%%%%%%%%%%%%%%%%%%%%%%%%%%%%%%%%%%%%%
\section{On Colinear Formations}
\label{sec:colinear_cases}
%%%%%%%%%%%%%%%%%%%%%%%%%%%%%%%%%%%%%%%%%%%%%%%%%%%%%%%%%%%%%%%%%%%%%

In this section we briefly discuss the colinear formations case, focusing on the bearing-preserving variations set.
According to {Definition}~\ref{MainDef:def:colinear}, a formation composed of $n\geq3$ {distinctly} placed agents is \textit{colinear} if all the agents are collinear, i.e., for any $k$-th component of the position vectors, $k \in \{1 \ldots d\}$, it exists  $c\in\mathbb{R}$ such that $p_i^k = c p_j^k$ for each pair $(v_i,v_j)$ of agents in the group. Under this hypothesis, we can observe that the shape uniqueness is guaranteed for a larger set of infinitesimal variations {with respect to} that described in the previous sections. Although this statement is valid independently on $\bar{\mathcal{D}}$, in the following we distinguish between the three cases previously treated.

\paragraph{$\bar{\mathcal{D}}=\mathbb{R}^{dn}$, $d\in\{2,3\}$}  for a formation composed of $n$ agents, controllable in $\mathbb{R}^d$, $d\in\{ 2,3\}$, and aligned along a certain direction identified by the unit vector {$\mathbf{w} \in \mathbb{S}^{d-1}$} the bearing measurements are collinear, namely {$
\bearingFunctionG{\config} = \diag{b_1\mathbf{I}_d \,\cdots \, b_m\mathbf{I}_d} \left({\mathbf 1}_m \otimes \mathbf{w}\right),
$ with $b_i\in{\mathbb R}$ for $i \in \{1 \ldots m\}$}. Thus, these are preserved despite the displacement of any agent along the direction specified by $\mathbf{w}$ and the translation of the whole formation in the subspace $\mathcal{W}$ of $\mathbb{R}^d$ orthogonal to {$\mathbf{w}$}. Hence the trivial variation set\footnote{Note that the uniform scaling of the formation corresponds to suitable (not equal) translations of all the agents along the direction $\mathbf{v} \in \mathbb{S}^{d-1}$.} coincides with $\mathcal{S}_t^d=\Span{\mathbf{I}_n \otimes {\mathbf{w}}, {\mathbf 1}_n \otimes \mathbf{W}}$ where $\mathbf{W} \in \mathbb{R}^{d \times (d-1)}$ is a matrix whose columns represent a basis for $\mathcal{W}$. Trivially, $\mathcal{S}_t^d$ has dimension $n+(d-1) > d+1 = \vert \mathcal{S}_t \vert$, with $\mathcal{S}_t$ as in~\eqref{Rd:eq:Trivial_Space}. 

\paragraph{$\bar{\mathcal{D}}=(\mathbb{R}^d \times \mathbb{S}^1)^n$, $d \in \{2,3\}$} for {an $n$-agent} formation acting in $(\mathbb{R}^d\times \mathbb{S}^1)^n$, the bearing measurements are retrieved in the local agents frame. Additionally, each agent has a (controllable) rdof  allowing rotations only around the direction of $\mathbf{v} \in \mathbb{S}^{2}$ when $d=3$. To analyze the colinear situation in which all the agents are aligned along the direction identified by the unit vector ${\mathbf{w}} \in \mathbb{S}^{d-1}$, it is necessary to distinguish between the following cases: $(i)$ $d=2$, or $d=3$ and {$\mathbf{v} \neq \mathbf{w}$}, $(ii)$ $d=3$ and {$\mathbf{v} = \mathbf{w}$}. For a colinear formation satisfying conditions $(i)$, the bearing measurements are preserved when the whole agents group translates along any direction in the $(d\!-\!1)$-dimensional subspace $\mathcal{W}\subseteq \mathbb{R}^d$ orthogonal to {$\mathbf{w}$}, when a coordinated rotation is performed according to the definition given in {Section}~\ref{sec:RdxSO(2)}, and also when any agent moves along the alignment direction. Therefore, the trivial variation set $\mathcal{S}_t^d$ is spanned by $n+(d-1)+1$ elements. In case $(ii)$, the dimension of $\mathcal{S}_t^d$ increases since the formation is not required to perform a coordinated rotation to preserve the bearings: also the rotation of any agent around the axis identified by {$\mathbf{v} = \mathbf{w}$} ensures the measurements maintenance. Hence, we get $\vert \mathcal{S}_t^d \vert = 2n+(d-1)$. Note that in both cases $(i)$ and $(ii)$  the trivial variation set has dimension grater {with respect to} non-colinear case for which $\vert \mathcal{S}_t \vert = d+2$ according~to~\eqref{RdxSO2:eq:Trivial_Space}.

\paragraph{$\bar{\mathcal{D}}=SE(3)^n$}  when the space of interest is $\bar{\mathcal{D}}=SE(3)^n$,  we figure out that for the colinear case in which the $n$ agents are aligned along the direction identified by  the unit vector {$\mathbf{w} \in \mathbb{S}^{2}$},  bearings are preserved when any agent translates or rotates along the direction of $\mathbf{v}$ and when the whole formation performs a translation or a coordinated rotation around any direction in the (two-dimensional) subspace $\mathcal{W} \subseteq \mathbb{R}^3$ orthogonal to {$\mathbf{w}$}. The trivial variation set has thus dimension $2n+4 > 7=\vert \mathcal{S}_t \vert$ with $\mathcal{S}_t$ as in~\eqref{RdxSO2:eq:Trivial_Space}.

In general, we can observe that for every {differential manifold} $\bar{\mathcal{D}}$  it occurs that $\vert \mathcal{S}_t^d \vert> \vert \mathcal{S}_t \vert$.

%%%%%%%%%%%%%%%%%%%%%%%%%%%%%%%%%%%%%%%%%%%%%%%%%%%%%%%%%%%%%%%%%%%%%

\section{Conclusions}

\label{sec:conclusions}
%%%%%%%%%%%%%%%%%%%%%%%%%%%%%%%%%%%%%%%%%%%%%%%%%%%%%%%%%%%%%%%%%%%%%

This work focuses on bearing rigidity theory applied to multi-agent systems whose elements are characterized by a certain number of both tdofs and rdofs. As original contribution, we  propose a general framework for the definition of the main rigidity properties without accounting for the specific controllable agents state domain. Moreover, we summarize the existing results about bearing rigidity theory for frameworks embedded in~$\mathbb{R}^d$, in $\mathbb{R}^d \times \mathbb{S}^1$ with $d \in \{2,3\}$ and in $SE(3)$. For each case, the principal definitions are provided and the infinitesimal rigidity property is investigated by deriving a necessary and sufficient condition based on the rigidity matrix rank. 
In addition, we provide a necessary and sufficient condition to check IBR property of a given system independently on its {differential manifold}. This arises from the derivation of a unified structure of the rigidity matrix that does not rest on the specific agents domain but exploits the fact that ${\mathcal{D}} \subseteq SE(3)$.

\section*{Acknowledgements}

We would like to thank M.~Pasquetti (Amminex Emissions Technology), A.~Franchi (University of Twente), %P.~Robuffo Giordano (CNRS, University of Rennes, Inria, IRISA) 
and S.~Zhao (Westlake University) for their % numerous suggestions and 
valuable comments.

\section*{Appendix}

\begin{thm}
\label{Rd:thm:Ker_B_K_for_n>=3}
For a {non-colinear} framework $(\mathcal{K}, \config)$ in $\mathbb{R}^{dn}$, $d \in \{ 2, 3\}$, it holds that $\Ker{\bearingMatrixK} = \trivialSpace$, or equivalently, $\rank{\bearingMatrixK}= dn-d-1$,
where $\trivialSpace$ corresponds to the trivial variation set~\eqref{Rd:eq:Trivial_Space}.   
\end{thm}

\begin{proof}
The bearing rigidity matrix associated to  $(\mathcal{K}, \config)$ is so that the $k$-th row block, corresponding to $e_k=(v_i,v_j) \in \mathcal{E}_\mathcal{K}$ with $i < j$\footnote{For undirected graph, it is always possible to choose a suitable edges labeling ensuring the desired requirement.}, has the following form where $\mathbf{0}_{p \times q}$ is the $(p \times q)$ zero matrix and $\mathbf{B}_{ij} =s_{ij} \mathbf{P}\left(\vectpNorm{i}{j}\right)\in \mathbb{R}^{d\times d}$
\begin{equation}
\label{Rd:eq:Bearing_Matrix_Construction}
\colvec{
\mathbf{0}_{d\times d(i-1)} & -\mathbf{B}_{ij} & \mathbf{0}_{d \times d(j-i-1)} & \ \ \mathbf{B}_{ij} & \mathbf{0}_{d \times d(n-j)} }.
\end{equation}
For $d=2$, $\mathbf{B}_{ij} =
s_{ij}^3\mathbf{r}_{ij}\mathbf{r}_{ij}^\top$,
with $\mathbf{r}_{ij}=[
p_{ij}^y \;  -p_{ij}^x
]^\top\in\mathbb{R}^{2}$, where $p_{ij}^{x}, p_{ij}^{y} \in \mathbb{R}$ are the (scalar) components of vector $\vectp{i}{j}  \in \mathbb{R}^2$ along the $x$-axis and $y$-axis of the global frame, respectively. Note that $\mathbf{B}_{ij}$ is neither zero nor full-rank, hence the $k$-row block~\eqref{Rd:eq:Bearing_Matrix_Construction} has unitary rank. For this reason, for each edge $e_k=\left(v_i, v_j \right)  \in \mathcal{E}_\mathcal{K}$  with $i < j$, we consider the next opportunely scaled version of~\eqref{Rd:eq:Bearing_Matrix_Construction}, %
\begin{equation}
\colvec{
\mathbf{0}_{1\times 2(i-1)} & -\mathbf{r}_{ij}^\top & \mathbf{0}_{1 \times 2(j-i-1)} & \mathbf{r}_{ij}^\top & \mathbf{0}_{1 \times 2(n-j)} },
%\mathbf{B}\left(n\right) =
%\colvec{
%\vdots & \vdots   & 	\vdots   & 	\vdots &	\vdots  \\
%\mathbf{0} & -\mathbf{r}^{\transposed}_{ij} & \mathbf{0} & \mathbf{r}^{\transposed}_{ij} & \mathbf{0} \\
%\vdots & 	\vdots   & 	\vdots   & 	\vdots &	\vdots  
%}\in \mathbb{R}^{(n-1)n\times 2n}. 
\end{equation}
obtaining the matrix $\mathbf{B}\left(n\right) \in \mathbb{R}^{((n-1)n/2)\times 2n}$. 
This has the same rank of $\bearingMatrixK$ but lower dimensions, so hereafter, we consider $\mathbf{B}\left(n\right)$ instead of  $\bearingMatrixK$ and we prove thesis by induction on the number $n$ of agents in the formation.\\
\textbf{Base case}:  $n=3$ 
\\ We aim at proving that $\rank{\mathbf{B}\left(3\right)}=3$.
To do so, observe that
\begin{align}
\mathbf{B}\left(3\right) = \colvec{
	-\mathbf{r}_{12}^\top & \mathbf{r}_{12}^\top & \mathbf{0}_{1 \times 2} \\
-\mathbf{r}_{13}^\top & \mathbf{0}_{1 \times 2} & \mathbf{r}_{13}^\top \\
\mathbf{0}_{1 \times 2} & -\mathbf{r}_{23}^\top & \mathbf{r}_{23}^\top }\in \mathbb{R}^{3\times 6}
\end{align}
is full-rank whether the agents are not all collinear.  
Because of non-colinear formation hypothesis the thesis is  proved.\\
\textbf{Inductive step} $ n=\bar{n}$:  Note that, given a set of $\bar{n}$ agents, for each subset containing $\bar{n}-1$ elements, it is possible to partition $\mathbf{B}\left(\bar{n}\right)$ so that {\begin{equation}
\label{Rd:eq:B_bar_n}
\mathbf{B}\left(\bar{n}\right) =
\colvec{  \mathbf{B}\left(\bar{n}-1\right) & \vectOnes{\eta} \otimes \mathbf{0}_2^\top \\[-0.25cm] \\\
\mathrm{diag}\{\mathbf{r}_{i\bar{n}}^\top\}_{i=1}^{(\bar{n}-1)} & \mathrm{col}\{\mathbf{r}_{i\bar{n}}^\top\}_{i=1}^{(\bar{n}-1)}}
\begin{array}{c}
{\footnotesize 1^{\text{st}}  \;\text{block}}  \\[-0.25cm] \\
{\footnotesize 2^{\text{nd}}  \; \text{block}}
\end{array}
\end{equation}}
where the first block has $\eta=\left(\bar{n}-1\right)\left(\bar{n}-2\right)/2$ rows related to the edges incident to the first $\bar{n}-1$ agents, while the second block has $\bar{n}$ rows related to the edges connecting the $\bar{n}$-th agent  with the first $\bar{n}-1$ agents.
For inductive hypothesis the thesis holds for $\bar{n}-1 \geq 3$, i.e., $\rank{\mathbf{B}\left(\bar{n}-1\right)}=2\bar{n}-5$,  thus, 
the first block of $\mathbf{B}\left(\bar{n}\right)$ in~\eqref{Rd:eq:B_bar_n} contains $2\bar{n}-5$ linearly independent rows. Moreover, there are at least two agents, for instance the $i$-th and $j$-th agent, that are not aligned with the $\bar{n}$-th agent, hence it does not exist $c\in\mathbb{R}$ such that $\mathbf{r}_{i\bar{n}} =c \mathbf{r}_{j\bar{n}}$ and the rows related to the edges $\left(v_i, v_{\bar{n}}\right)$, and $ \left(v_j, v_{\bar{n}}\right)$ are linearly independent {with respect to} the rows of the first block. $\mathbf{B}\left(\bar{n}\right)$ has thus at least $2\bar{n}-3$ linearly independent rows, and, since  $\rank{\mathbf{B}\left(\bar{n}\right)}\leq 2\bar{n}-3$ for {Lemma}~4 in~\cite{2016-ZhaoZelazo-RigidityRd}, then it must be  $\rank{\mathbf{B}\left(\bar{n}\right)}= 2\bar{n}-3$ concluding the proof for the case $d=2$.\\
When $d=3$, the matrix $\mathbf{B}_{ij}$ in~\eqref{Rd:eq:Bearing_Matrix_Construction} turns out to be
\begin{equation}
\label{Rd:eq:B_ij_d3}
\mathbf{B}_{ij} \! =
s_{ij}^3 \! \colvec{
\left(p_{ij}^y\right)^2 \! \! + \!\left( p_{ij}^z \right)^2\! \! & -p_{ij}^x p_{ij}^y \! \! & -p_{ij}^x p_{ij}^z \! \! \\
\rule{0pt}{3ex}   
-p_{ij}^yp_{ij}^x \! \!  & 	\left(p_{ij}^x\right)^2 \! \!+ \! \left( p_{ij}^z \right)^2 &  	-p_{ij}^y p_{ij}^z \! \! \\
\rule{0pt}{3ex}   
-p_{ij}^z p_{ij}^x \! \! & -p_{ij}^z p_{ij}^y \! \! &  \left(p_{ij}^x \right)^2 \!  \! + \! \left( p_{ij}^y \right)^2  \\
} 
\end{equation}
where $p_{ij}^{x}, p_{ij}^{y}, p_{ij}^{z} \in \mathbb{R}$ are the (scalar) components of vector $\vectp{i}{j}  \in \mathbb{R}^3$ along the $x$-axis, $y$-axis, and $z$-axis of the $\mathscr{F}_W$, respectively. The proof for this case thus follows the same inductive reasoning performed for $d=2$.
\end{proof}

\bibliographystyle{IEEEtran}

% DO NOT ERASE THE NEXT LINE,
% ONLY COMMENT IT AND DECOMMENT THE NEXT-NEXT, IF YOU NEED
%\bibliography{bibAlias,bibMain,bibNew,bibAF}
\bibliography{./bibCustom}
\begin{IEEEbiography}[{\includegraphics[width=1in,height=1.25in,clip,keepaspectratio]{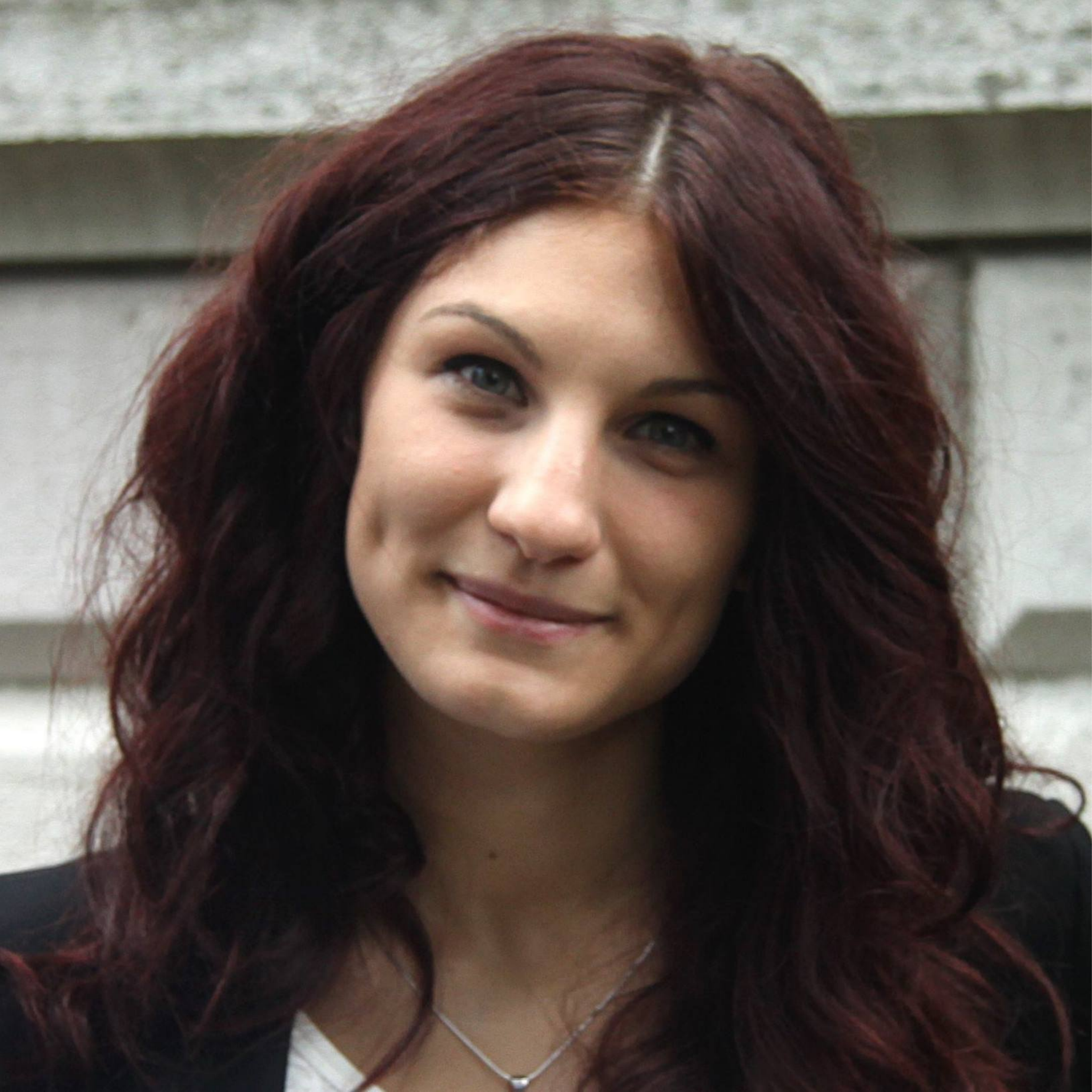}}]
 {\bf Giulia Michieletto} (M'18) received the M.S. (2014) and the Ph.D. (2018) degrees from the University of Padova, Italy, where she is currently an Assistant Professor with the Department of Management and Engineering. From March 2016 to February 2017, she was a Visiting Researcher at LAAS-CNRS, Toulouse, France.  From February 2018 to November 2019, she was a post-doc fellow with the SPARCS group at University of Padova, Italy.  Her main research interests include multi-agent systems modeling and control with a special regard to networked formations of aerial vehicles and nano satellites.
 \end{IEEEbiography}
% \vspace{-16pt}
\begin{IEEEbiography}[{\includegraphics[width=1in,height=1.25in,clip,keepaspectratio]{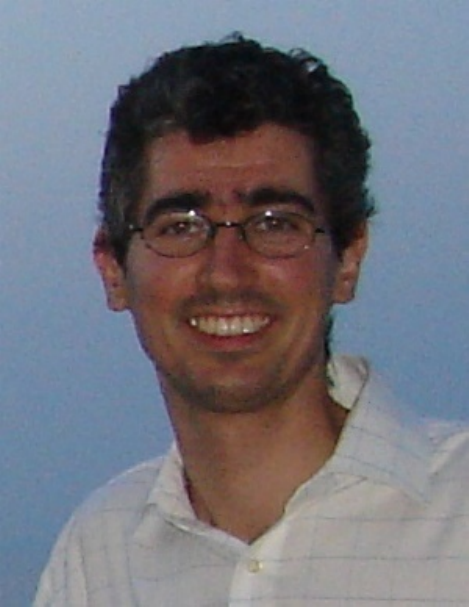}}]
 {\bf Angelo Cenedese} (M'11) received the M.S. (1999) and the Ph.D. (2004) degrees from the University of Padova, Italy, where he is currently an Associate Professor with the Department of Information Engineering and leader of the SPARCS research group. He has held several visiting positions at the UKAEA-JET laboratories (UK), the UCLA Vision Lab (CA-USA), the F4E European Agency (Spain). His research interests include system modeling, control theory and its applications, sensor and actuator networks, multi agent systems. On these subjects, he has published more than 170 papers and holds three patents.
  \end{IEEEbiography}
%\vspace{-16pt}
\begin{IEEEbiography}[{\includegraphics[width=1in,height=1.25in,clip,keepaspectratio]{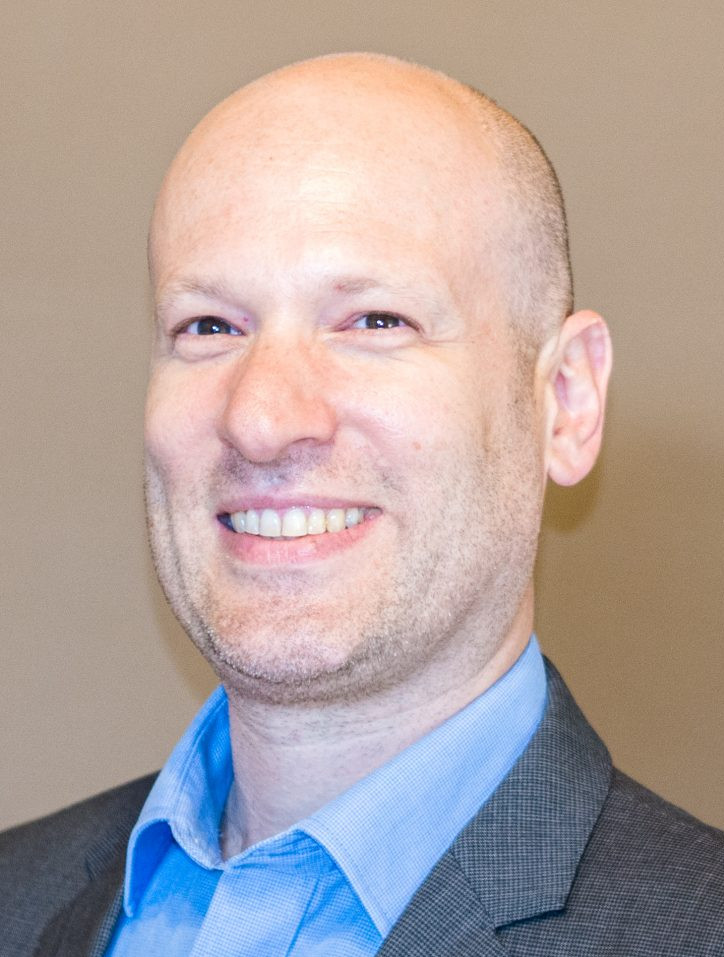}}]
 {\bf Daniel~Zelazo} is an Associate Professor of Aerospace Engineering at the 
 Technion - Israel Institute of Technology. He received his BSc. (1999) and M.Eng 
 (2001) degrees in Electrical Engineering from the Massachusetts Institute of 
 Technology. In 2009, he completed his Ph.D. from the University of Washington 
 in Aeronautics and Astronautics. From 2010-2012 he served as a post-doctoral 
 research associate and lecturer at the Institute for Systems Theory \& 
 Automatic Control in the University of Stuttgart. His research interests 
 include topics related to multi-agent systems, optimization, and graph theory.
 \end{IEEEbiography}

\end{document}